\documentclass[journal]{IEEEtran}

\usepackage{paralist,tikz,graphicx,amsmath, amsfonts,amssymb,amsthm,overpic,enumitem, latexsym,epsfig,xcolor,rotating,paralist,times,float,subcaption,algorithm,verbatim,fancybox}
\usetikzlibrary{fit,positioning}

  \newcommand{\beq}{\begin{equation}}
  \newcommand{\eeq}{\end{equation}}
\newcommand{\argmax}{\operatornamewithlimits{arg\,max}}
\newcommand{\argmin}{\operatornamewithlimits{arg\,min}}

\newcommand{\threshold}{\gamma}

\newcommand{\horizon}{N}

\newcommand{\reals}{{\rm I\hspace{-.07cm}R}}
\newcommand{\state}{x}
\newcommand{\obs}{y}
\newcommand{\anoise}{\epsilon}
\newcommand{\belief}{\pi}

\newcommand{\Belief}{\Pi}

\newcommand{\obspace}{\mathcal{Y}}

\newcommand{\tPhi}{\tilde{\Phi}}

\newcommand{\bound}{\bar{\kalmancov}}
\newcommand{\ebound}{\bar{\lambda}}
\newcommand{\carrier}{\omega_c}
\newcommand{\cdf}{F}
\newcommand{\ccdf}{\bar{F}_M}

\newcommand{\waveparam}{\theta}
\newcommand{\snr}{\eta}
\newcommand{\waveparamone}{\theta_1}
\newcommand{\waveparamtwo}{\theta_2}
\newcommand{\envelope}{s}
\newcommand{\precconstraint}{p_*}

\newcommand{\hstate}{\hat{\state}}

\newcommand{\upperbound}{\bar{\response}}

\newcommand{\evalue}{\lambda}
\newcommand{\ARE}{\operatorname{\mathcal{A}}}

\newcommand{\Sig}{K}
\newcommand{\cost}{C}
\newcommand{\budget}{B}
\newcommand{\budgetg}{g}

\newcommand{\dataset}{D}
\newcommand{\obsdataset}{\mathcal{D}_\text{obs}}
\newcommand{\response}{\beta}
\newcommand{\obsresponse}{\bar{\response}}
\newcommand{\Obsresponse}{\boldsymbol{\obsresponse}}

\newcommand{\Response}{\boldsymbol{\response}}
\newcommand{\norm}[1]{\lVert#1\rVert}

\newcommand{\prob}{\mathbb{P}}
\newcommand{\probe}{\alpha}
\newcommand{\pnoise}{\zeta}
\newcommand{\nprobe}{\bar{\probe}}
\newcommand{\Probe}{\boldsymbol{\probe}}
\newcommand{\Obsprobe}{\boldsymbol{\nprobe}}

\newcommand{\probedim}{m}

\newcommand{\ole}{\stackrel{\text{defn}}{=}}
\newcommand{\nresponse}{\obsresponse}

\newcommand{\sindx}{s}
\newcommand{\tindx}{t}

\makeatletter
\newcommand{\pushright}[1]{\ifmeasuring@#1\else\omit\hfill$\displaystyle#1$\fi\ignorespaces}
\makeatother

\newcommand{\lagrange}{\lambda}

\newcommand{\statespace}{\mathcal{X}}
\newcommand{\statedim}{X}
\newcommand{\obsdim}{Y}

\newcommand{\pdf}{p}
\newcommand{\statem}{A}
\newcommand{\snoise}{w}
\newcommand{\onoise}{v}

\newcommand{\Anoise}{\boldsymbol{\anoise}}
\newcommand{\Pnoise}{\boldsymbol{\pnoise}}
\newcommand{\lambdat}{\lambda^o}
\newcommand{\ut}{u^o}
\newcommand{\obsm}{C}

\newcommand{\snoisecov}{Q}
\newcommand{\onoisecov}{R}
\newcommand{\kalmancov}{\Sigma}
\newcommand{\normal}{N}
\newcommand{\uniform}{\operatorname{Unif}}
\newcommand{\p}{\prime}
\newcommand{\diag}{\operatorname{diag}}
\newcommand{\trace}{\operatorname{Tr}}

\newcommand{\filter}{T}

\newcommand{\utility}{U}

\newtheorem{theorem}            {Theorem}
\newtheorem*{result*}            {Result}

\newtheorem{definition}         [theorem]{Definition}

\newtheorem{lemma}              [theorem]{Lemma}

  {\popQED\end{theorem}}

\newcommand{\dtime}{n}

\usepackage[hyphens]{url}
\usepackage{tikz,pgfplots}
\usetikzlibrary{shapes,arrows.meta}
\usetikzlibrary{decorations.pathreplacing,
                positioning, 
                quotes,fit}
                \usetikzlibrary{chains,shapes.multipart}
\usetikzlibrary{shapes,calc,fit}
\usetikzlibrary{automata,positioning}

\tikzset{
    block/.style={rectangle, draw, line width=0.5mm, black, text width=8em, text centered,
                 minimum height=2em},
               line/.style={draw, -latex}}

\tikzset{
    block2/.style={rectangle, draw, line width=0.2mm, black, text centered,
                 minimum height=2em},
    line/.style={draw, -latex}}      

\tikzset{
    blocka/.style={rectangle, draw, line width=0.5mm, black, text width=4.5em, text centered,
                 minimum height=1em},
               line/.style={draw, -latex}}

\begin{document}

\title{Identifying Cognitive Radars -  Inverse Reinforcement Learning  using Revealed Preferences}

\author{Vikram~Krishnamurthy, {\em Fellow IEEE}  and Daniel Angley and  Robin Evans, {\em Fellow IEEE}  and William Moran \\
  Manuscript dated \today
  \thanks{Vikram Krishnamurthy is  with the School of Electrical and Computer Engineering, Cornell University.
    Email: vikramk@cornell.edu.  Daniel Angley and Robin Evans and Bill Moran are with the Department of Electrical and Electronic Engineering, University of Melbourne.

  This research was
funded by  Air Force Office of Scientific Research grant FA9550-18-1-0007
through the Dynamic Data Driven Application Systems Program.}}

\maketitle

\begin{abstract}
  We consider an  inverse reinforcement learning  problem involving
``us'' versus an ``enemy''  radar equipped with a Bayesian tracker.  By observing
the emissions of the  enemy  radar, how can we identify if the radar is cognitive (constrained utility maximizer)?
 Given  the observed sequence of   actions taken by the enemy's radar, we consider three problems: (i) Are the enemy radar's actions (waveform choice, beam scheduling) consistent with constrained utility maximization? If so how can we estimate the  cognitive radar's utility  function that is consistent with its actions.   We formulate and solve the problem in terms of the spectra (eigenvalues) of the state and observation noise covariance matrices, and the algebraic Riccati equation.  (ii) How to construct a statistical  test for detecting a cognitive radar (constrained utility maximization) when we observe the radar's actions in noise or the radar observes our probe signal in noise?  We propose a statistical detector with a tight Type-II error bound.
(iii) How can we optimally probe (interrogate) the enemy's radar by choosing our state   to minimize the Type-II  error of detecting if the radar is deploying an economic rational strategy, subject to a constraint on the Type-I detection error? We present a stochastic optimization algorithm to optimize our probe signal.
The main analysis framework used in this paper is that of revealed preferences from microeconomics. 

\end{abstract}


\IEEEpeerreviewmaketitle

\section{Introduction}

Cognitive radars  \cite{Hay06} use  the perception-action cycle of cognition to  sense the environment, learn from it relevant information about the target and the background, then adapt the radar sensor to optimally satisfy the needs of their mission. A
crucial element of a cognitive radar  is optimal adaptivity: based on its tracked estimates, the radar   adaptively optimizes 
the waveform, aperture, dwell time and revisit rate. In other words, a cognitive radar is a constrained utility maximizer.

This paper  is motivated by the  next logical step, namely, \textit{inverse cognitive radar}. From the intercepted emissions of an enemy's radar:  (i) How can we identify if  the enemy's radar is cognitive? That is, are the enemy radar's actions consistent with optimizing a   utility function (equivalently,  is the  radar's behavior rational in an economics sense). If so how to  estimate the  cognitive radar's utility  function that is consistent with its actions?  (ii) How to construct a statistical detection test for utility maximization when we observe the enemy's radar's actions in noise and the enemy radar observes our probe signal in noise?
 (iii) How can we optimally probe the enemy's radar by choosing our state  to minimize the Type-II  error of detecting if the enemy radar is deploying an economic rational strategy, subject to a constraint on the Type-I detection error? 

    The central theme of this paper involves an adversarial signal processing/inverse reinforcement learning problem\footnote{Inverse reinforcement learning (IRL)  \cite{NR00} seeks to estimate the utility function of a decision system by observing its input output dataset. The revealed preferences framework considered here is more general since it identifies if the behavior is consistent with a utility function and then estimates a set of utility functions that rationalize the dataset. Also  revealed preferences involves active learning in that the observer probes the system whereas classical IRL is passive with no  probe signal.}  comprised of   ``us'' and an ``adversary''. Figure \ref{fig:schematica} displays the schematic setup.
``Us'' refers to a drone/UAV  or electromagnetic signal that probes an ``adversary'' cognitive  radar system.
The adversary's cognitive radar  estimates our kinematic coordinates using a Bayesian tracker and then adapts its mode (waveform, aperture, revisit time) dynamically using feedback control  based on sensing our kinematic state (e.g.\ position and velocity of drone).
At each time $\dtime$ our kinematic state can be viewed\footnote{In Section \ref{sec:kf} we give specific examples of how the kinematic state and radar actions are mapped to probe $\probe_\dtime$ and response $\response_\dtime$, respectively.} as a probe vector $\probe_\dtime \in \reals_+^\probedim$ to the radar. We  observe  the  radar's response $\response_\dtime \in \reals^\probedim$. Given the time series of probe vectors and responses, $\{(\probe_\dtime,\response_\dtime),\dtime=1,\ldots,\horizon\}$
is it possible to say if the radar is ``rational'' (in an economics-based sense)? That is, does there exist a utility function $\utility(\response)$ that the radar  is maximizing to generate its response $\response_\dtime$ to our probe input $\probe_\dtime$? How can we estimate such a utility function to predict the future behavior of the cognitive radar?

    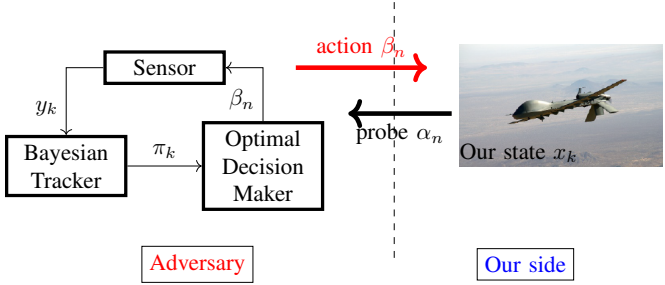
\begin{figure}[h]
            {\resizebox{8cm}{!}{
\begin{tikzpicture}[node distance = 1cm, auto]
    \node [blocka] (BLOCK1) {Sensor};
    \node [blocka, below of=BLOCK1,right of=BLOCK1,node distance=1.5cm] (BLOCK2) {Optimal Decision \\ Maker};
    \node [blocka, below of=BLOCK1,left of=BLOCK1,node distance=1.5cm] (BLOCK3) {Bayesian Tracker};

    \draw[<-] (BLOCK1) -| node[left,pos=0.8]{$\response_\dtime$}  (BLOCK2)  ;
    \draw[->] (BLOCK1.west) -|   node[left,pos=0.8]{$\obs_k$} (BLOCK3);

    \draw[->](BLOCK3) --  node[above]{$\belief_k$} (BLOCK2);

    \node[draw=none,fill=none] at (4.5,-1.5) (drone) {\includegraphics[bb=0 0 0 0,scale=0.4]{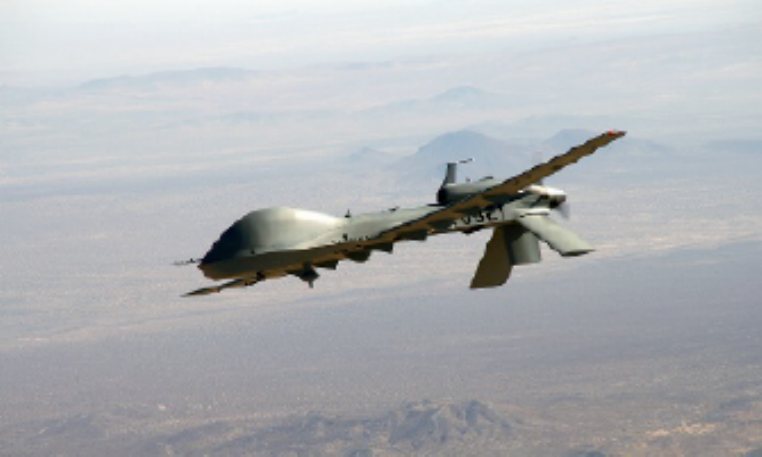}};
     \node[text width=2cm] at (5.5,-1.3) {{Our state $\state_k$}};

    \draw[->,color=red,line width=2pt] (2,0)   -- node[above]{action $\response_\dtime$}(4,0);
    \draw[->,line width=2pt] ([yshift=0.8cm]drone.west)   --   node[below]{probe $\probe_\dtime$} (2.8,-0.7);
    \node[draw] at (5.5,-3.0) {{\color{blue}Our side}};
    \node[draw] at (0.5,-3.0) {{\color{red}Adversary}};
    \draw [dashed] (3.5,1) -- (3.5,-3);
   \end{tikzpicture}} }

\caption{Schematic of Adversarial Inference Problem. Our side is a drone/UAV or electromagnetic signal that probes the  enemy's cognitive radar system. $k$ denotes a fast time scale and $n$ denotes a slow time scale. Our state $\state_k$, parameterized by $\probe_\dtime$ (purposeful acceleration maneuvers), probes the adversary radar. Based on the noisy observation $\obs_k$ of our state, the enemy radar  responds with action~$\response_\dtime$. Our aim is to determine if the enemy radar is economic rational, i.e.,  is its response  $\beta_\dtime$  generated by constrained optimizing a utility function?}
\label{fig:schematica}
\end{figure}

 \subsection{Revealed Preferences and Afriat's Theorem} \label{subsec:afriat} Nonparametric detection of utility maximization behavior is the central theme in the area of revealed preferences in microeconomics; which dates back to Samuelson in 1938  \cite{Sam38}.
\begin{definition}[\cite{Afr67,Afr87}]\label{eq:utility_maximizer}
A system is  a {\em utility maximizer} if for every probe $\probe_\dtime\in \reals_+^\probedim$, the  response $\response_\dtime \in \reals^\probedim$ satisfies
\begin{equation}
  \response_\dtime\in \argmax_{\probe_\dtime^\p \response \leq 1}\utility(\response) 
\label{eq:utilitymaximization}
\end{equation}
where $\utility(\response)$ is a {\em monotone} utility function.
 \end{definition}

In economics, $\probe_\dtime$ denotes the price vector and  $\response_\dtime$ the consumption vector. Then $\probe_\dtime^\p \response \leq 1$ is a natural budget constraint\footnote{As discussed below, the budget constraint $\probe_\dtime^\p \response \leq 1$ is without loss of generality, and can be replaced by $\probe_\dtime^\p \response \leq c$ for any positive constant $c$.}  for a consumer with 1 dollar. Given a dataset of price and consumption vectors, the aim is to determine
   if the consumer is a utility maximizer (rational) in the sense of (\ref{eq:utilitymaximization}).  
 
 The key result in revealed preferences is the 
following remarkable theorem  
due to Afriat; see~\cite{Die12,Afr87,Afr67,Var12,Var83} for extensive expositions.

\begin{theorem}[Afriat's Theorem~\cite{Afr67}] Given a data set
\begin{equation}
 \dataset=\{(\probe_\dtime,\response_\dtime), \dtime\in \{1,2,\dots,\horizon\}\},
\label{eqn: response data}
\end{equation}
 the following statements are equivalent: 
	\begin{compactenum}
	\item The system is a utility maximizer and there exists a monotonically increasing,\footnote{By definition, an economics-based utility function is monotone increasing, i.e.,  $\response_1 \leq \response_2$ (elementwise) implies $\utility(\response_1) \leq \utility(\response_2)$, and we will use this definition throughout the paper. Monotone is a special case of a more general class of {\em locally non-satiated} utility functions \cite{MWG95}.  In this paper, we use monotone and local non-satiation interchangeably. Afriat's theorem was originally stated for a non-satiated utility function.} continuous,  and concave utility function by satisfies (\ref{eq:utilitymaximization}). 
		\item For $u_t$ and $\lagrange_t>0$ the following set of inequalities (called Afriat's inequalities)  has a feasible solution:
			\begin{equation}
				u_\sindx-u_\tindx-\lambda_\tindx \probe_\tindx^\p (\response_\sindx-\response_\tindx) \leq 0 \; \forall \tindx,\sindx\in\{1,2,\dots,\horizon\}.\
				\label{eqn:AfriatFeasibilityTest}
			\end{equation}
		\item Explicit  monotone and concave utility
                  functions that rationalize the dataset by satisfying (\ref{eq:utilitymaximization}) are given by:
			\begin{equation}
				\utility(\response) = \underset{\tindx\in \{1,2,\dots,\horizon\}}{\operatorname{min}}\{u_\tindx+\lambda_\tindx \probe_\tindx^\p(\response-\response_\tindx)\}
				\label{eqn:estutility}
                              \end{equation}
                              where $u_\tindx$ and $\lambda_\tindx$ satisfy the linear  inequalities (\ref{eqn:AfriatFeasibilityTest}).
                      \item The data set $\mathcal{D}$ satisfies the Generalized Axiom of Revealed Preference (GARP), namely for any $\tindx \leq \horizon$,
                        \beq \probe_t^\p \response_t \geq \probe_t^\p \response_{t+1} \quad \forall t\leq k-1 \implies \probe_k^\p  \response_k \leq \probe_k^\p  \response_{1}.  \label{eq:garp}\eeq
	\end{compactenum}
\label{thm:AfriatTheorem}
\end{theorem}

Afriat's theorem tests for economics-based rationality; its  remarkable property is that it gives a {\em necessary and sufficient condition} for a system  to be a utility maximizer based on the system's input-output response. 
The feasibility of the set of inequalities (\ref{eqn:AfriatFeasibilityTest}) can be checked using a linear programming solver; alternatively GARP (\ref{eq:garp})  can be checked  using Warshall's algorithm with $O(\horizon^3)$ computations~\cite{Var06}~\cite{Var82}. 
A utility function consistent with the data can be constructed\footnote{As pointed out in Varian's influential paper \cite{Var82}, another remarkable feature of Afriat's theorem is that if the dataset can be rationalized by a monotone utility function, then it can be rationalized
by a continuous, concave, monotonic utility function. Put another way,   continuity and  concavity cannot be refuted with  a finite dataset.\label{foot:varian}} using~\eqref{eqn:estutility}.

The recovered utility using~\eqref{eqn:estutility}  is not unique; indeed  any positive monotone increasing transformation of~\eqref{eqn:estutility} also satisfies Afriat's Theorem; that is, the utility function constructed is ordinal. This is the reason why the budget constraint $\probe_\dtime^\p \response \leq 1$ is without loss of generality; it can be scaled by an arbitrary positive  constant and  Theorem \ref{thm:AfriatTheorem} still holds.  In signal processing terminology, Afriat's Theorem can be viewed as set-valued system identification of an \emph{argmax} system; set-valued since (\ref{eqn:estutility}) yields a set of utility functions that rationalize the finite dataset $\dataset$.

\subsection{Objectives} 
In this paper, \textit{our working assumption is that a cognitive radar satisfies economics-based rationality}; that is, a cognitive radar is  a constrained utility maximizer in the sense of (\ref{eqn:estutility}) with possibly a nonlinear budget constraint.
 The main objectives of the paper involve answering:

 \textbf{1. Test for Utility Maximization -- Spectral Revealed Preferences}: 
The first question is: Does a radar satisfy economics based rationality, i.e., is its  action $\response_\dtime$ consistent with optimizing a utility function $\utility$? 
By observing how the enemy's radar switches ambiguity function and waveforms to track a target, or how the radar schedules its beam between targets, is there a utility function that rationalizes the radar's behavior? Notice that a key requirement in Afriat's theorem is a budget constraint.
How to formulate a useful budget constraint for a radar?
A key idea in this paper is to  formulate  linear and nonlinear  budget constraints for a radar in terms of the tracking error covariance where $\probe_\dtime$ and $\response_\dtime$ are the spectra of the state and observation noise matrices (as will be justified in Section \ref{sec:kf}) associated with a Kalman filter tracker.  Specifically, the linear budget constraint is used in Sec.\ref{sec:kf} for waveform design, and Sec.\ref{sec:beam} for beam scheduling,  while a non-linear budget constraint is used to formulate utility maximization in terms of the spectrum of covariance matrices.
From a practical point of view, such spectral revealed preferences yield
constructive estimates of  the radar's utility function,  and so we can predict (in a Bayesian sense) its  future actions. 

2. {\bf Cognition Detection in Noise}: If   the radar's response $\response_\dtime$ or probe signal $\probe_\dtime$ is observed in noise, then violation of Afriat's theorem could be either due to measurement noise or the absence of utility maximization. We will construct a statistical detection test to decide if the radar is a utility maximizer. The hypothesis test yields a tight bound for the Type-I errors.

 \textbf{3. Optimal Probing}.  Given the  detector in the above objective, what choice of probe signal yields the smallest Type-II error in detecting if the radar is a utility maximizer, subject to maintaining the Type-I error within a specified bound? We construct a stochastic gradient algorithm that estimates our optimal probe sequence.

\subsection{Context and Literature}
The above objectives are  fundamentally different to the {\em model-centric} theme  used in the signal processing  literature where
one postulates an objective function (typically convex) and then proposes  optimization algorithms.  In contrast the revealed
preference approach is {\em data centric} - given a dataset,  we wish to determine if it is consistent with utility maximization.
Specifically, Sections \ref{sec:kf}  and \ref{sec:beam} below  discuss how  revealed preferences can be used as a systematic method  to identify utility maximization in cognitive radars.

Regarding the literature, in the context of revealed preferences we already mentioned \cite{Var12,Var85,Var83,Die12,Afr67}. A nonlinear budget version extension was developed
in \cite{FM09} which we will exploit in our spectral revealed preferences setup  in Section \ref{sec:kf}. A stochastic detector for utility maximization given noisy measurements of the probe or response
is studied in \cite{FW05,JE09} and we will use these results in Section \ref{sec:noise}. Our earlier work \cite{KH12,AK17} consider utility estimation in adversarial signal processing and social network applications. As mentioned above, revealed preferences are more general than inverse reinforcement learning \cite{NR00}.

Cognitive radars
\cite{Hay12}   use stochastic control and optimal resource allocation to adapt their waveform \cite{KE94},  beam allocation \cite{KD09}, aperture, and service requests.
In the last decade there have been numerous  works in adaptive/cognitive radar and radar resource management; see  \cite{CKH09,KD07,Kri16} and references therein.
What has not been studied is: by listening to a radar, can one identify if the radar is a utility maximizer, and if so, estimate its utility function. This is the subject of our paper.  Below we will use revealed preferences to  identify radars that optimize their waveforms and their beam allocation. 
   Our  aim is to give a necessary and sufficient condition  to identify if a radar is cognitive, estimate its utility function,  construct a statistical detector for  utility maximization's when the radar is observed in  noise (or the radar observes us in noise) and then adaptively optimize our probe signal to minimize the classification error of the detector. Although not discussed in this paper, once we can detect cognitive behavior and estimate the radar's utility function, we can predict future actions possibly spoof/jam the radar.


Finally, this paper builds on our recent work \cite{KR19,MIC19} in Bayesian adversarial signal processing where the aim is to reconstruct the posterior distribution of the enemy's tracker given its actions. While  \cite{KR19,MIC19} deal with inverse Bayesian estimation problems, the focus here is on the more general problem of detecting constrained utility maximization in a non-parametric setting.

\section{Cognitive Radar Response Model} \label{sec:generalmodel}

The setup involves two time scales.
Let $k=1,2,\ldots$ denote discrete time (fast time scale) and $\dtime=1,2,\ldots$ denote epoch (slow time scale).
Our probe signal is $\probe_\dtime$, the radar's response action is $\response_\dtime$ and our measurement of this action is~$\obsresponse_\dtime$. 


The  model of ``us'' interacting with  the cognitive radar  has
the following dynamics, see Figure \ref{fig:setup}:
\begin{equation}
\begin{split}
  \text{our state: }    \state_k &\sim  \pdf_{\probe_\dtime}(\state | \state_{k-1}), \quad \state_0 \sim \belief_0 \\
  \text{ radar action: }     \response_\dtime &\in \argmax_{\probe_\dtime^\p \response \leq 1} \utility(\response)  \\
\text{ radar observation: }     \obs_k  &\sim  \pdf_{\response_\dtime}(y | x_k)\\
\text{ radar tracker: }    \belief_k &= \filter(\belief_{k-1}, \obs_k)\\
\text{observed action: } \obsresponse_\dtime &= \response_\dtime + \anoise_\dtime
  \end{split} \label{eq:model}
\end{equation}

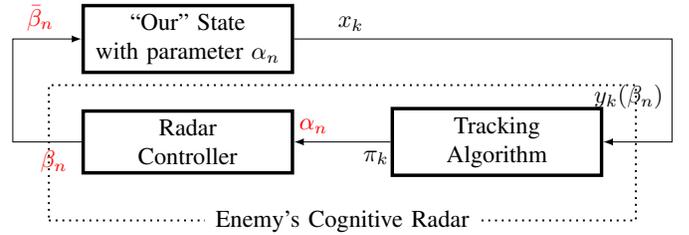
\begin{figure}
\resizebox{9cm}{!}{ 
\begin{tikzpicture}[node distance = 4.5cm, auto]
  \node [block] (BLOCK1) {``Our'' State \\ with parameter $\probe_\dtime$ };
    \node [block, below of=BLOCK1,node distance=1.5cm] (BLOCK2) {Radar\\ Controller};
    \node[block, right of=BLOCK2] (grad) {Tracking\\ Algorithm};
 
    \path [line] (BLOCK1.east) --++ (5.5cm,0cm) node[pos=0.15,above]{$\state_k$}
    |-   node[pos=0.28,left]{$\obs_k(\response_n)$}    (grad);
    
     \path [line] (BLOCK2.west) --++ (-1cm,0cm)  node[red, pos=0.4,below]{$\response_\dtime$}  |-     node[red, pos=0.7,above]{$\obsresponse_\dtime$}     (BLOCK1.west);
     \path [line] (grad) --   node[pos=0.15,below]{$\belief_k$}    node[red, pos=0.8,above]{$\probe_\dtime$}  (BLOCK2);
     \node[draw, thick, dotted, inner sep=3ex, yshift=-1ex,
     fit=(BLOCK2) (grad)] (box) {};
     \node[fill=white, inner xsep=1ex] at (box.south) {Enemy's Cognitive Radar};
   \end{tikzpicture}
 }
 \caption{Interaction of our dynamics with adversary's cognitive radar. The cognitive radar is comprised of a Bayesian tracker and a radar controller. Based on the time series $(\probe_\dtime, \response_\dtime), \dtime = 1,\ldots,\horizon$, or more generally
   $(\probe_\dtime, \obsresponse_\dtime), \dtime = 1,\ldots,\horizon$, our goal is to determine of the radar controller is a utility optimizer in the sense of (\ref{eq:utilitymaximization}). }
 \label{fig:setup}
\end{figure}

Let us explain the notation in (\ref{eq:model}):
 $\pdf(\cdot)$ denotes a generic conditional  probability density function (or probability mass function),  $\sim$ denotes distributed according to, and
\begin{compactitem}
\item $\state_k\in \statespace$ is our Markovian state with  transition kernel $\pdf_{\probe_\dtime}$ and prior $\belief_0$ where $\statespace$ denotes the state space.
  \item Our dynamics are determined by the control probe signal
  $\probe_\dtime$ which evolves on the slow time scale.   Our probing of the enemy radar is  performed  via purposeful maneuvers. We will model $\probe_\dtime$ using two different levels of abstraction. In Sec.\ref{sec:kf} we  use $\probe_\dtime$ to model the state maneuver noise covariance matrix. In Sec.\ref{sec:beam} we will work at  a higher level of abstraction and use $\probe_n$ to model the covariance at the enemy's Kalman tracker (which is a  deterministic  function of the state  covariance matrix).

 \item Based on optimizing a utility function  $\utility$ (which is unknown to us)  of  the predicted target statistic (e.g.\ covariance of the target's estimate) in epoch $\dtime$, the enemy radar chooses an action $\response_\dtime$.  It is here that actual  tracker structure determines the response.
\item $\obs_k\in \obspace$ is the radar's noisy observation of our state $\state_k$; with observation likelihoods $\pdf_{\response_\dtime}(\obs|\state)$.
  Here $\obspace$ denote the observation space.
  \item The observation at the radar depends on its action $\response_\dtime$, which evolves on the slow time scale.
This reflects the fact that the cognitive radar adapts (optimizes) its receive and transmit functionalities.
  For example, it adapts its matched filter to its transmit waveform.
 
    \item $\belief_k = \pdf(x_k| \obs_{1:k})$ is the radar tracker's belief (posterior)  of our state $\state_k$ where $\obs_{1:k}$ denotes the sequence  $\obs_1,\ldots,\obs_k$. The tracking functionality  $\filter(\cdot)$ in (\ref{eq:model}) is the classical Bayesian optimal filtering update formula  \cite{Kri16}
       \beq  \filter(\belief,\obs)(\state) = \frac{
    \pdf_{\response}(\obs|\state) \,\int_\statespace 
     \pdf_{\probe}(\state | \zeta) \, \belief(\zeta) \,d\zeta}
   {\int_\statespace
 \pdf_{\response}(\obs|\state) \,\int_\statespace 
     \pdf_{\probe}(\state | \zeta) \, \belief(\zeta) \,d\zeta
     d \state}
  \label{eq:belief}
  \eeq
  Note that the cognitive radar's tracker update depends on the both the probe $\probe_n$ and response $\response_n$ signals. 
Let $\Belief$ denote the  space of all such beliefs. When the state space
$\statespace$ is Euclidean space, then $\Belief$ is a function space comprising the space of density functions; if $\statespace$ is finite, then $\Belief$ is  the unit $
\statedim-1$ dimensional simplex of $\statedim$-dimensional probability mass functions.
\item $\obsresponse_\dtime$ denotes our noisy measurement of the radar's action
    \end{compactitem}

    The above model substantially generalizes the adversarial signal processing model in our recent paper \cite{KR19} since  now the state, observation and tracker dynamics are controlled by the probe and response signals. As mentioned earlier, \cite{KR19} focused on  Bayesian estimation of posterior $\belief_k$; in comparison this paper  addresses  the deeper problem of
\begin{compactenum} \item
  determining if the radar response signal is consistent with constrained utility maximization, \item  estimating utility $\utility(\response)$ subject to the signal processing constraints in (\ref{eq:model}).
\end{compactenum}

To summarize, a cognitive radar chooses its action to maximize a utility function, and adapts its receiver to the optimized action. In terms of Afriat's theorem, we will use the following economics-based interpretation: the  probe signal
  $\probe_\dtime$ is the price the radar pays for  tracking our target,
  while $\response_\dtime$ is the amount of resources (consumption) the radar spends on the target
  at epoch $\dtime$.
  We will justify this price/consumption framework  in economics (budget constraint) terms at two levels of abstraction: waveform adaptation for a single target  (Sec.\ref{sec:kf}) and beam scheduling amongst multiple targets  (Sec.\ref{sec:beam}). We will also show how a nonlinear budget constraint arises in the context of the spectrum of the state covariance matrix.

{\em Remark. Game-theoretic setting:} This paper assumes the radar responds to our probe in an optimal way. In a more sophisticated game-theoretic setting, a radar is aware that we are probing it, and may deliberately use a sub-optimal response to confuse us. Identifying if our strategy and the radar's strategy are
  consistent with play from the equilibrium of a game is a difficult problem and not considered here; see \cite{HKA16} for partial results in the special case of potential games.

\section{Waveform Adaptation: Spectral Revealed Preferences to Test for Cognitive Radar}
\label{sec:kf}

Waveform adaptation is perhaps one of the most important functionalities of a  cognitive radar. 
A cognitive radar adapts its waveform by adapting  its  ambiguity function.  Our  aim is to identify  such cognitive behavior of the enemy's radar when it  deploys a Bayesian  filter as a physical level tracker. For concreteness, in this section  we assume that the enemy's cognitive radar uses a  Kalman filter tracker. Also since the probe and response signal evolve on a slow time scale $\dtime$ (described below) we assume that both the radar and us (observer) have perfect measurements of probe $\probe_\dtime$ and response $\response_\dtime$.

Our  working assumption is that a cognitive radar satisfies economics-based rationality; that is, it adapts its waveform by maximizing  a utility function in the sense  of (\ref{eqn:estutility}) with  a possibly nonlinear  budget constraint.
A key requirement in Afriat's Theorem \ref{thm:AfriatTheorem} is the budget constraint. In economics, such a constraint is obvious since it specifies the total available resources of the decision maker.
\textit{How to formulate  useful budget constraints for waveform adaptation?} Our key idea here is  to  formulate  linear and nonlinear  budget constraints  in terms of the Kalman filter  error covariance where $\probe_\dtime$ and $\response_\dtime$ are the spectra (eigenvalues)  of the state and covariance noise matrices of the state space model.

\subsection{Waveform Adaptation by Cognitive Radar} \label{sec:waveform}
Suppose a radar adapts its waveform  while tracking a  target (us) using a Kalman filter.
Our probe input comprises purposeful maneuvers that modulate the spectrum (vector of eigenvalues) of the state noise covariance matrix. The radar responds with an optimized waveform which modulates the spectrum of the observation noise covariance matrix.
By observing the radar's signals,  how can we test the radar for economic rationality?  

\subsubsection{Linear Gaussian Target Model and Radar Tracker}
Linear Gaussian dynamics for a target's kinematics \cite{LJ03} and
 linear Gaussian measurements at the radar are widely assumed as a useful  approximation \cite{BLK08}. Accordingly, consider the 
following   special case of model (\ref{eq:model}) with linear Gaussian dynamics and measurements:
\beq \label{eq:lineargaussian}
\begin{split}
\state_{k+1} &= \statem\, \state_k  + \snoise_k(\probe_\dtime), \quad \state_0 \sim \belief_0 \\
\obs_k &= \obsm\, \state_k + \onoise_k(\response_\dtime)
\end{split}
\eeq
Here  $\state_k \in \statespace = \reals^\statedim$ is ``our'' state with
initial density $\belief_0 \sim \normal(\hat{\state}_0,\kalmancov_0)$,
 $\obs_k \in \obspace = \reals^\obsdim$ denotes the cognitive radar's observations,
 $\snoise_k\sim \normal(0,\snoisecov(\probe_\dtime))$,
 $\onoise_k \sim \normal(0,
\onoisecov(\response_\dtime))$
and 
  $\{\snoise_k\}$,  
  $\{\onoise_k\}$ are mutually independent  i.i.d.\ processes.
  When $\state_k$ denotes respectively, the x,y,z position and velocity components of the target (so $\state_k \in \reals^6$) then
 \beq \statem_{6 \times 6} = \diag\bigl[ \begin{bmatrix}  1 & T \\ 0 & 1
  \end{bmatrix}, \begin{bmatrix}  1 & T \\ 0 & 1
  \end{bmatrix}, \begin{bmatrix}  1 & T \\ 0 & 1
  \end{bmatrix}  \bigr]  \label{eq:target} \eeq   where $T$ is the sampling interval.
  Recall $k$ indexes the fast time scale while  $\dtime$ indexes the slow time scale. 
  
  In (\ref{eq:lineargaussian})  we  explicitly  indicate the dependence of the state noise covariance $\snoisecov$
on our probe signal $\probe_\dtime$ and 
the  observation noise covariance $\onoisecov$ on the radar's response signal $\response_\dtime$.  These are justified as follows. When the radar controls its ambiguity function, in effect it controls the measurement noise covariance $\onoisecov$. Of course, this come as a cost:
reducing the  observation noise covariance of a target results in  increased visibility of the radar (and therefore higher threat) or increased covariance of other targets. 
Note that the radar  re-configures its receiver (matched filter) each time it chooses a waveform; (\ref{eq:lineargaussian}) abstracts all the physical layer aspects of the radar response into the observation noise covariance $ \onoisecov(\response_\dtime)$.

Our probing of the enemy radar is  performed  via purposeful maneuvers by  modulating our  state covariance matrix  $\snoisecov$  in (\ref{eq:lineargaussian}) by $\probe_\dtime$.
For example, in a classical linear Gaussian state space model used in target tracking \cite{BLK08}, our probe $\probe_\dtime$ parametrizes  the state noise covariance $\snoisecov(\probe_\dtime)$  which models acceleration maneuvers of our drone.

 Based on observation sequence $\obs_1,\ldots,\obs_k$, the tracking functionality in the radar computes the  posterior $$\belief_k = \normal(\hstate_k,\kalmancov_k)$$ where $\hstate_k$ is the conditional mean
  state   estimate and $\kalmancov_k$ is the covariance. These are computed by the classical Kalman filter:
\beq
  \begin{split}
\kalmancov_{k+1|k} &=  \statem  \kalmancov_{k} \statem^\p  +  \snoisecov(\probe_\dtime)  \\
\Sig_{k+1} &= \obsm \kalmancov_{k+1|k} \obsm^\p + \onoisecov(\response_\dtime)
\\
{\hstate}_{k+1} &=  \statem\,  {\hstate}_k  + 
\kalmancov_{k+1|k} \obsm^{\p}  \Sig_{k+1}^{-1} 
(\obs_{k+1} - \obsm \, \statem\,  {\hstate}_k )
\\
\kalmancov_{k+1} &=
\kalmancov_{k+1|k} -  
\kalmancov_{k+1|k} \obsm^{\p}  \Sig_{k+1}^{-1} 
\obsm \kalmancov_{k+1|k} 
\end{split}
\label{eq:kalman}
\eeq
Under the assumption that the model parameters in (\ref{eq:lineargaussian}) satisfy $[\statem,\obsm]$ is detectable and $[\statem,\sqrt{\snoisecov}]$ is stabilizable, 
the asymptotic predicted covariance $\kalmancov_{k+1k|k}$ as
$k\rightarrow \infty$ is the unique non-negative definite solution of  the \textit{algebraic Riccati equation} (ARE):
\begin{multline}
     \ARE(\probe,\response,\kalmancov) \ole  \\
    - \kalmancov + \statem  \big(\kalmancov -  
\kalmancov \obsm^{\p}  \left[ \obsm \kalmancov \obsm^\p + \onoisecov(\response) \right]^{-1}
\obsm \kalmancov \big)  \statem^\p  +  \snoisecov(\probe) = 0
\label{eq:are}
\end{multline}
where $\probe_\dtime$ and $\response_\dtime$ are the probe and response signals
of the radar at epoch $\dtime$. Note
$ \ARE(\probe,\response,\kalmancov) $ is a symmetric $ \reals^{\probedim \times \probedim} $ matrix.
Since $\kalmancov$ is parametrized by $\probe,\response$, we 
write the solution of the ARE at epoch $\dtime$ as  $\kalmancov_\dtime^*(\probe,\response)$.

\subsection{Effect of  waveform design on observation noise covariance} To give a precise structure to the radar dynamics, this section  summarizes how the  observation noise covariance $\onoisecov(\response)$
in (\ref{eq:lineargaussian}) depends on the radar waveform. The details involve maximum likelihood estimation involving the radar ambiguity function and can be found in \cite{Van68,KE94}.
Below:
\begin{compactitem}
\item  $c$ denotes the speed of light (in free space),
  \item $\carrier$ denotes the carrier frequency,
\item  $\waveparam$ is an adjustable parameter in the waveform,
\item $\snr $ is the signal to noise ratio at the radar.
\item $j = \sqrt{-1}$ is the unit imaginary number.
\item $\tilde{s}(t)$ is the complex envelope of the waveform.
  \item $\response$ is the vector of  eigenvalues of $\onoisecov$ (in Section \ref{sec:nonlinearbudget} below) or $\onoisecov^{-1}$ (in Section  \ref{sec:linear} below).
\end{compactitem}

We now describe 3 waveforms and their resulting observation noise
covariance matrices $\onoisecov(\response)$; see \cite{KE94} for details.

{\em (i) Triangular Pulse - Continuous Wave}
\begin{equation} \label{eq:triangular}
  \begin{split}
\tilde{\envelope}(t)&=\left\{\begin{array}{ll}{\sqrt{\frac{3}{2 \waveparam}}\left(1-\frac{|t|}{\waveparam}\right)} & {-\waveparam<t<\waveparam} \\ {0} & {\text { otherwise }}\end{array}\right.
\\
  \onoisecov(\response) &=\begin{bmatrix} {\frac{c^{2} \waveparam^{2}}{12 \snr}} & {0} \\ {0} & {\frac{5 c^{2}}{2 \carrier^{2} \waveparam^{2} \snr}} \end{bmatrix}
\end{split}
\end{equation}

{\em (ii) Gaussian Pulse - Continuous Wave}
\begin{equation} \label{eq:gaussianpulse}
  \begin{split}
\tilde{\envelope}(t)&=\left(\frac{1}{\pi \waveparam^{2}}\right)^{1 / 4} \exp \left(\frac{-t^{2}}{2 \waveparam^{2}}\right)
\\
\onoisecov(\response)&=\left[\begin{array}{cc}{\frac{c^{2} \waveparam^{2}}{2 \snr}} & {0} \\ {0} & {\frac{c^{2}}{2 \carrier^{2} \waveparam^{2} \snr}}\end{array}\right]
\end{split}
\end{equation}

{\em (iii) Gaussian Pulse - Linear Frequency Modulation  chirp}
\begin{equation}
  \begin{split}
\tilde{\envelope}(t)&=\left(\frac{1}{\pi \waveparamone^{2}}\right)^{1 / 4} \exp \left(-\left(\frac{1}{2 \waveparamone^{2}}-j \waveparamtwo \right) t^{2}\right)
\\
\onoisecov(\response)&=\left[\begin{array}{cc}{\frac{c^{2} \waveparamone^{2}}{2 \snr}} & {\frac{-c^{2} \waveparamtwo \waveparamone^{2}}{\carrier \snr}} \\ {\frac{-c^{2} \waveparamtwo \waveparamone^{2}}{\carrier \snr}} & {\frac{c^{2}}{\carrier^{2} \snr}\left(\frac{1}{2 \waveparamone^{2}}+2 \waveparamtwo^{2} \waveparamone^{2}\right)}\end{array}\right]
\end{split}
\end{equation}
  
To summarize, by adapting its waveform parametrized by $\response$ (vector
of eigenvalues),
the radar can change the noise covariance $\onoisecov(\response)$.  Below we will use the response $\response_\dtime$ to construct revealed preference tests for cognition.

\subsection{Testing for Cognitive Radar: Spectral Revealed Preferences with  Linear Budget} \label{sec:linear}
We now show that Afriat's theorem (Theorem \ref{thm:AfriatTheorem}) can be used to determine if a radar is cognitive. The assumption here is that the utility function $\utility(\response)$ maximized by the radar is a monotone function (unknown to us) of the predicted covariance of the target.
Our main task is to formulate and justify a linear budget constraint $\probe_\dtime^\p \response \leq 1$ in Afriat's theorem.

Specifically,
suppose
\begin{compactenum}
\item Our probe  $\probe_\dtime$ that characterizes our maneuvers, is the vector of eigenvalues of the positive definite matrix
  $\snoisecov$
\item The radar response $\response_\dtime$ is the vector of eigenvalues of the positive definite matrix~$\onoisecov^{-1}$.
\end{compactenum}
Then the
cognitive radar chooses  its  waveform parameter $\response_\dtime$  at each slow time epoch $\dtime$  to maximize a utility $\utility(\cdot)$:  \beq \response_\dtime\in \argmax_{\probe_\dtime^\p \response \leq 1}\utility(\response) \label{eq:radaropt} \eeq
where $\utility$ is a monotone increasing function of $\response$.

Then Afriat's theorem   (Theorem \ref{thm:AfriatTheorem}) can be used to detect utility maximization and construct a utility function that rationalizes the response of the radar.
Recall that  the 1 in the right hand side of the budget $\probe_\dtime^\p \response \leq 1 $ can be replaced by any non-negative constant.

It only remains to justify the linear budget constraint $\probe_\dtime^\p \response \leq 1$
in (\ref{eq:radaropt}).  The $i$-th component of $\probe$, denoted as $\probe(i)$,  is the incentive for considering  the $i$-th mode of the target;  $\probe(i)$ is proportional to the signal power. The $i$-th component of $\response$ is the amount of resources (energy)  devoted by the radar to this $i$-th mode; a higher $\response(i)$ (more resources) results in a smaller  measurement noise covariance, resulting in higher accuracy of measurement by the radar.
So $\probe_\dtime^\p \response$ measures  the signal to noise ratio (SNR) and the budget constraint $\probe_\dtime^\p \response \leq 1$ is a bound on the SNR.
A rational  radar  maximizes a utility  $\utility(\response)$ that is monotone increasing in the accuracy (inverse of noise power) $\response$.  However, the radar  has limited resources and can only expend sufficient resources to ensure that the precision (inverse  covariance) of all modes is at most  some pre-specified precision  $\bound ^{-1}$ at each epoch $\dtime$.
We can then justify the linear budget constraint as follows:

\begin{lemma} \label{lem:linear}The linear budget constraint
  $\probe_\dtime^\p \response \leq 1 $ implies that solution of the ARE (\ref{eq:are}) satisfies ${\kalmancov^*_\dtime}^{-1} (\probe_\dtime,\response) \preceq \bound^{-1}$
  for some symmetric positive definite matrix $\bound^{-1}$.
\end{lemma}

The proof of Lemma \ref{lem:linear} follows straightforwardly using the information Kalman filter formulation \cite{AM79}, and showing that ${\kalmancov^*}^{-1}$ is increasing  in $\response$. Afriat's theorem requires that the constraint
$\probe_\dtime^\p \response \leq 1 $  is active at $\response = \response_\dtime$. This holds in our case
since ${\kalmancov^*}^{-1}$ is increasing  in~$\response$.

To summarize, we can use Afriat's theorem (Theorem \ref{thm:AfriatTheorem}) with $\probe_\dtime$ as the spectrum of $\snoisecov$ and $\response_\dtime$ as the spectrum of $\onoisecov$,  to test a cognitive radar for utility maximization (\ref{eq:radaropt}). Moreover, Afriat's theorem  constructs a set of  utility functions (\ref{eqn:estutility}) that  rationalize the decisions of the radar.


\subsection{Testing for Cognitive Radar:   Spectral Revealed Preferences with  Nonlinear Budget Constraint} \label{sec:nonlinearbudget}

This section constructs a method to identify cognitive radars by generalizing
Afriat's theorem to a nonlinear budget constraint.
The nonlinear  budget constraint (nonlinear in $\response$)  emerges naturally from the covariance of the Kalman filter tracker, namely, the ARE (\ref{eq:are}).  We 
 use this together with an extension of  Afriat's theorem to test if a radar satisfies economic  rationality.
The interpretation of the probe and response are different (in some sense ``opposite'') to that of the linear case:
\begin{compactenum}
  \item The probe vector $\probe_\dtime\in \reals^\probedim$ is the vector of eigenvalues of  $\snoisecov^{-1}$.
\item The radar response 
  $\response_\dtime \in \reals^\probedim_+$ is the vector of eigenvalues of $\onoisecov$.
\item Define $\evalue(\kalmancov_\dtime^*(\probe_\dtime,\response))$  as   the largest eigenvalue of
  $\kalmancov_\dtime^*(\probe_\dtime,\response)$ where $\kalmancov^*$ is the solution of the ARE (\ref{eq:are}).
\end{compactenum}

With the above definitions, our aim is to test if the radar's response $\response$  satisfies economics-rationality:
\begin{equation} \begin{split}
    \response_\dtime &= \argmax_{\response} \utility(\response), \\
    \text{ subject to:   } &
    \evalue\big( \kalmancov^*_\dtime(\probe_\dtime,\response) \big) \leq \ebound,\quad   \response \leq \upperbound_\dtime
  \end{split}
  \label{eq:radarnonlinear}\eeq
  Since there is no natural ordering of eigenvalues, our assumption is that
  $\utility(\response)$ is a symmetric function\footnote{Examples of  symmetric functions include trace, determinant, nuclear norm, etc. The assumption of symmetry is only required when we choose $\response $ to be the vector of eigenvalues since there is no natural ordering of the eigenvalues in terms of the ordering of the elements of the matrix. Specifically, Theorem~\ref{thm:nonlinear} does not require $\utility(\response)$ to be a symmetric function of $\response$.} of $\response$.
Here $\kalmancov^*_\dtime$ is the solution of the ARE (\ref{eq:are}) at epoch $\dtime$, and
$\ebound \in \reals_+$,  $\upperbound_\dtime \in
\reals^\probedim_+$ are user-specified parameters. Note that the constraint  $\response \leq \upperbound_\dtime $  holds elementwise.

\subsubsection{Economics-based Rationale for  Utility and Nonlinear Budget constraint (\ref{eq:radarnonlinear})}
The economics-based rationale for the utility (\ref{eq:radarnonlinear}) is as follows: The $i$-th component of
$\probe$, denoted as $\probe(i)$,  is the price the radar pays for devoting resources to the $i$-th mode of the
target. Since $\probe(i)$ is inversely proportional to the signal power; so a higher
$\probe(i)$ implies a more expensive mode to track, implying
that the enemy radar needs to 
allocate
more resources to the $i$-th mode.
The radar's response for the $i$-th mode is $\response(i)$; this reflects the
cost incurred by the radar for estimating mode $i$. A rational radar aims to minimize its
total effort $\cost(\response)$ where cost  $\cost(\response)$ decreases with
$\response$ since  choosing a
waveform that results in a larger
observation noise
variance requires less effort.  Equivalently, the radar seeks to maximize a utility function
$\utility(\response) = - \cost(\response)$ where $\utility(\response)$ is increasing
with $\response$.

We now discuss the nonlinear budget 
constraint  $\evalue\big( \kalmancov^*_\dtime(\probe_\dtime,\response) \big) \leq \ebound $ in (\ref{eq:radarnonlinear}) together with $\response
\leq \upperbound_\dtime$.
The radar seeks to minimize total effort $\cost(\response)$  subject to maintaining the  inaccuracy of all modes (covariance $\kalmancov_\dtime$) to be smaller than some pre-specified covariance $\bound $.
Clearly, a sufficient condition
is that $\evalue\big( \kalmancov^*_\dtime(\probe_\dtime,\response) \big) \leq \ebound $.
But 
for revealed preferences involving nonlinear budgets, we need the following (see Theorem \ref{thm:nonlinear} below):  The constraint $\evalue\big( \kalmancov^*_\dtime(\probe_\dtime,\response) \big) \leq \ebound $ in (\ref{eq:radarnonlinear}) needs to be   active at $\response_n$. This is straightforwardly ensured by choosing
$\ebound$ as
\beq \ebound \in [0, \lambda_L] , \quad \text { where }
\lambda_L = \evalue(\kalmancov_\dtime^*(\probe_\dtime,\upperbound_\dtime) )  
\label{eq:ALC} \eeq
That is,  $\lambda_L $ is the largest eigenvalue of the unique solution $\kalmancov_L$ of the
ARE $ \ARE(\probe_\dtime,\upperbound_\dtime,\kalmancov) = 0$.
The constraint  (\ref{eq:ALC}) says that  the enemy's Bayesian tracker cannot perform worse in covariance  than that of  the worst case observation noise covariance $\onoisecov(\upperbound)$.
i.e., $\kalmancov^*\leq \kalmancov_L$ (positive definite ordering).

{\em Remark}.
In the special case when the constraint $\response \leq \upperbound$ is omitted, then $ \kalmancov_L$ is the solution of the
algebraic Lyapunov equation
\beq \kalmancov=  \statem \kalmancov \statem^\p + \snoisecov(\probe) \label{eq:ale} \eeq The constraint  (\ref{eq:ALC}) then says that  the enemy's Bayesian tracker cannot perform worse than the optimal predictor (which has infinite observation noise). Of course, when  $\statem$ is specified as in (\ref{eq:target}), since all the eigenvalues of $\statem$ are 1, the solution of the algebraic Lyapunov equation is not finite. 

We can now  justify the nonlinear budget for a cognitive radar equipped with a Kalman filter tracker as follows:

\begin{lemma} Consider the nonlinear  budget constraint
   $\evalue\big( \kalmancov^*_\dtime(\probe_\dtime,\response) \big) \leq \ebound $ in  (\ref{eq:radarnonlinear}) with user defined parameter $\ebound$ satisfying (\ref{eq:ALC}). Then  the solution of the ARE (\ref{eq:are}) satisfies $\kalmancov_\dtime^*(\probe_n,\response) \preceq \bound$, for  any choice of  symmetric positive definite matrix~$\bound \preceq \kalmancov_L$.
\end{lemma}

\subsubsection{Revealed Preference for Nonlinear Budget} Having formally justified the nonlinear budget constraint   $\evalue\big( \kalmancov^*_\dtime(\probe_\dtime,\response) \big) \leq \ebound $ in (\ref{eq:radarnonlinear}), we now state the main
revealed preference test
\cite{FM09} which generalizes
Afriat's theorem to nonlinear budgets. The result below provides an explicit test for a cognitive radar and constructs a set of utility functions that rationalizes the  decisions $\{\response_n\}$ of the cognitive radar.

\begin{theorem}[Test for rationality with nonlinear budget \cite{FM09}]
  Let $\budget_{\dtime} = \{\response \in \reals^\probedim_+| \budgetg_\dtime(\response) \leq 0 \}$ with $\budgetg_\dtime: \reals^\probedim \rightarrow \reals$ an increasing, continuous function and $\budgetg_\dtime(\response_\dtime) = 0$ for $\dtime =1,\ldots,\horizon$. Then the following conditions are equivalent:
  \begin{compactenum}
  \item There exists a monotone continuous utility function $\utility$ that rationalizes
    the data set $\{\response_\dtime,\budget_\dtime\}, \dtime=1,\ldots,\horizon$.
    That is
    $$ \response_\dtime = \argmax_{\response} \utility(\response), \quad
    \budgetg_\dtime(\response_\dtime) \leq 0 $$
  \item The data set $\{\response_\dtime,\budget_\dtime\}, \dtime=1,\ldots,\horizon$ satisfies GARP:
    $$ \budgetg_t(\response_j) \leq \budgetg_t(\response_t) \implies \budgetg_j(\beta_t) \geq 0 $$
     \item For $u_t$ and $\lagrange_t>0$ the following set of inequalities has a feasible solution:
			\begin{equation}
				u_\sindx-u_\tindx-\lambda_\tindx \budgetg_\tindx(\response_\sindx)  \leq 0 \quad \forall \tindx,\sindx\in\{1,2,\dots,\horizon\}.\
				\label{eq:nonlinearFeasibilityTest}
			\end{equation}
     \item With $u_\tindx$ and $\lambda_\tindx$ defined in (\ref{eq:nonlinearFeasibilityTest}),
an explicit  monotone continuous  utility function that rationalizes the data set is given by:
			\begin{equation}
				\utility(\response) = \underset{\tindx\in \{1,2,\dots,\horizon\}}{\operatorname{min}}\{u_\tindx+\lambda_\tindx\, \budgetg_\tindx(\response) \} 
				\label{eq:nonlinearutility}
			\end{equation}

  \end{compactenum}
  \label{thm:nonlinear}
\end{theorem}

{\em Remarks}: (i) Clearly   Afriat's theorem (Theorem \ref{thm:AfriatTheorem})  is a  special case of Theorem \ref{thm:nonlinear} where $\budgetg_\dtime(\response) = \probe_\dtime^\p(\response- \response_\dtime)$. But unlike Afriat's theorem, the constructed utility function is not necessarily  concave.
\\ (ii) Just like Afriat's theorem,   (\ref{eq:nonlinearFeasibilityTest}) comprises of linear inequalities in $ u_\tindx, \lambda_\tindx$. So feasibility can be checked using an LP solver.

We now show that the nonlinear radar  budget constraint  in (\ref{eq:radarnonlinear}), (\ref{eq:ALC})
satisfies the properties of Theorem \ref{thm:nonlinear} with
\beq  \budgetg_\dtime (\response) = 
\evalue\big( \kalmancov^*_\dtime(\probe_\dtime,\response) \big) -  \ebound  \eeq 
First, clearly $\kalmancov^*(\probe,\response)$ is increasing in $\response$ and is a continuous function
of $\response$, and so is $\evalue\big(
\kalmancov^*(\probe,\response) \big)$. Second Theorem \ref{thm:nonlinear} requires the constraint  to
be active at $\response_\dtime$. This 
follows  since  $\evalue\big(
\kalmancov^*_\dtime(\probe_\dtime,\response) \big )  $ is increasing in~$\response$ and due to (\ref{eq:ALC}).

{\em Summary}: By choosing the probe signal  $\probe$ as the spectrum of $\snoisecov^{-1}$ and the response signal $\response$ as the spectrum of  $\onoisecov$,  we can use the nonlinear budget Theorem \ref{thm:nonlinear} to test a cognitive radar for utility maximization. We can then  construct explicit utility functions (\ref{eq:nonlinearutility}) that rationalize the decisions of the radar in terms of waveform adaptation.

\section{Beam Allocation: Revealed Preference Test}\label{sec:beam}
This section constructs a test for cognitivity of a radar that switches its beam adaptively between targets. We work at a higher level of abstraction than the previous section and consider  multiple targets. At this higher level of abstraction, we view each component $i$  of the probe signal
  $\probe_\dtime(i)$ as the trace of the precision matrix (inverse covariance) of target $i$. Note that the precision matrix is a deterministic function of the maneuver covariance of target $i$, and in the previous subsection we used this maneuver covariance as the probe signal. In comparison, we now use the trace of the precision of each target in our probe signal -- this allows us to consider multiple targets.

Suppose a radar adaptively switches its beam between
$\probedim$ targets where these $\probedim$ targets are controlled by us. As in (\ref{eq:lineargaussian}), on the fast time scale indexed by $k$, each target $i$ has linear Gaussian dynamics  and the enemy  radar obtains linear Gaussian measurements:
\beq \label{eq:lineargaussian2}
\begin{split}
\state^i_{k+1} &= \statem\, \state^i_k  + \snoise^i_k, \quad \state_0 \sim \belief_0 \\
\obs^i_k &= \obsm\, \state^i_k + \onoise^i_k, \quad i=1,2,\ldots,\probedim
\end{split}
\eeq
Here $ \snoise^i_k\sim \normal(0,\snoisecov_\dtime(i))$,
 $\onoise_k^i \sim \normal(0,
 \onoisecov_\dtime(i))$.
 We assume that both $\snoisecov_\dtime(i)$ and $\onoisecov_\dtime(i)$ are known to us and the enemy.

 As in previous sections, $n$ indexes the slow time scale and $k$ indexes the fast time scale.
 The enemy's radar tracks our $\probedim$ targets using  Kalman filter trackers.
 The fraction of time the radar allocates to each target $i$ in epoch $\dtime$ is $\response_\dtime(i)$. 
The price the radar pays for each target $i$ at the beginning of epoch $\dtime$ is the trace of the  predicted {\em precision} of target $i$. Recall that this is  the trace
of the inverse of the predicted  covariance   at  epoch $\dtime$ using the Kalman predictor\footnote{Since $\statem$ has all its eigenvalues at 1, we cannot use the algebraic Lyapunov equation (\ref{eq:ale}) as it does not have  bounded solution.}
\beq  \probe_\dtime(i) = \trace(\kalmancov^{-1}_{\dtime|\dtime-1}(i)), \quad i =1 ,\ldots, \probedim
\label{eq:probe_beam}
\eeq
The predicted covariance $\kalmancov_{\dtime|\dtime-1}(i)$ is a deterministic function of  the maneuver covariance $\snoisecov_\dtime(i)$ of target $i$.  So the probe  
  $\probe_\dtime(i)$ is a signal that we can choose, since it is a deterministic function of  the maneuver covariance $\snoisecov_\dtime(i)$ of target $i$.  Unlike the previous section  where the  spectrum of the probe matrix was chosen as the probe vector, here we abstract the target's covariance by the  trace $\probe_\dtime(i)$.  Note also that  the observation noise covariance $\onoisecov^i_\dtime$ depends on the enemy's radar response $\response_\dtime(i)$, i.e.,  the fraction of time allocated to target $i$.
We assume that each target $i$ is equipped with a radar detector and can estimate\footnote{If we impose a probabilistic structure on the estimates, then the resulting problem of statistical detection of a utility maximizer (stochastic revealed preferences)  is discussed in Section \ref{sec:noise}.} the fraction of  time $\response_\dtime(i)$ the enemy's radar devotes to it.

Given the time series $\probe_\dtime, \response_\dtime$, $\dtime = 1,\ldots,\horizon$,
our aim is to detect if the enemy's radar is cognitive. We assume that 
a cognitive radar optimizes its  beam allocation as follows:
\begin{equation}
  \begin{split}
    \response_\dtime &=     \argmax_\response \utility(\response)  \\
  \text{ s.t. } &   \response^\p \probe_\dtime \leq \precconstraint, 
  \end{split} \label{eq:beam}
\end{equation}
where $\utility(\cdot)$ is the enemy radar's utility function (unknown to us)  and $\precconstraint \in \reals_+$ is a pre-specified average precision of  all $\probedim$  targets.

The economics-based rationale  for the budget constraint is natural: For targets that are cheaper (lower precision $\probe_\dtime(i)$), the radar  has incentive to  devote more time $\response_\dtime(i)$. However, given its resource constraints,
the radar can  achieve at most an  average precision of $\precconstraint$ over all targets.

Note that the setup  (\ref{eq:beam}) is directly amenable  to   Afriat's Theorem \ref{thm:AfriatTheorem}.
Thus  (\ref{eqn:AfriatFeasibilityTest}) can be used to test if the radar satisfies utility maximization in its beam scheduling (\ref{eq:beam}) and also estimate the set of utility functions~(\ref{eqn:estutility}).
Furthermore (as in Afriat's theorem) since the utility is ordinal, $\precconstraint$ can be chosen as 1 without loss of generality (and therefore does not need to be known by us).

\section{Detecting Cognitive Radars in a Noisy Setting} \label{sec:noise}
Thus far we have discussed revealed preference based methods  to identify cognitive radars when the radar response is measured perfectly by us.
Afriat's theorem (Theorem~\ref{thm:AfriatTheorem}) and its generalization to nonlinear budgets (Theorem~\ref{thm:nonlinear}) assumes perfect observation of the probe and response. 
However, when the response (e.g.\  enemy's radar waveform) is measured in noise by us, or  the probe signal (e.g.\ our maneuver) is measured in noise by the enemy,  violation of the inequalities in Afriat Theorem could be either due to measurement noise or absence of utility maximization (economic rationality). In this section we give two statistical detection tests for utility maximization and characterize the Type-I and Type-II errors of the detector. We give the tightest possible Type-I error bound. This section also sets the stage for
Sec.\ref{sec:adapt} where the probe signal is optimizes to minimize the Type-II error of the detector.

\subsection{Detecting Cognitive Radar given Noisy Response}
Suppose we observe  the response $\response_\dtime$ of the enemy's  radar in  additive noise $\anoise_\dtime$ as
\begin{equation}
	\nresponse_\dtime = \response_\dtime + \anoise_\dtime.
	\label{eqn:noisemodel}
      \end{equation}
      Here $\anoise_\dtime$ are $\probedim$-dimensional  random variables that are possibly correlated but functionally  independent of $\response_\dtime$. 
          As an example, consider the setup of Section \ref{sec:beam}   where a cognitive radar  allocates its beam between multiple targets. Each target $i$ equipped with a radar detector obtains a noisy
      estimate of  the fraction of time $\response_\dtime(i)$ the enemy radar devotes to it.

Given the noisy data set 
\begin{equation}
	\obsdataset= \left\{\left(\probe_\dtime,\obsresponse_\dtime \right): \dtime \in \left\{1,\dots,\horizon \right\}\right\},
	\label{eqn:noisydataset}
\end{equation}
from the enemy radar, how can we detect if is cognitive?
Let
\begin{compactitem} \item
$H_0$ denote the null hypothesis that the data set $\obsdataset$ in~\eqref{eqn:noisydataset} satisfies utility maximization. 
\item $H_1$ denote the alternative hypothesis that the data set does not satisfy utility maximization.
\end{compactitem}
There are two possible sources of error:
\begin{align}
	\text{\bf Type-I errors:}   &\text{\hspace{1.6mm}Reject $H_0$ when $H_0$ is valid.} \nonumber\\
	\text{\bf Type-II errors:}  &\text{\hspace{1.6mm}Accept $H_0$ when $H_0$ is invalid.}
	\label{eqn:hypothesis}
\end{align}
Given $\obsdataset$, we  propose the following statistical test to determine if the enemy radar is a utility maximizer~\eqref{eq:utilitymaximization}:
\begin{equation}
\boxed{\phantom{\text{\hspace{0.2mm}}}\int\limits_{\Phi^*(\Obsresponse)}^{+\infty}f_M(\psi)\mathrm{d}\psi \overset{H_0}{\underset{H_1}{\gtrless}} \threshold}. 
	\label{eqn:Statistical_Test}
\end{equation}
In the statistical test (\ref{eqn:Statistical_Test}): \\ (i) $\threshold$ is the ``significance level'' of the test. 
\\(ii) The ``test statistic'' $\Phi^*(\Obsresponse)$, with ${\Obsresponse}=\left[\obsresponse_1,\obsresponse_2,\dots,\obsresponse_\horizon\right]$ is the solution of the following constrained optimization problem:
\begin{equation}
\begin{array}{rl}
\min & \Phi \\
\mbox{s.t.} & u_{\sindx}-u_{\tindx}- \lambda_\tindx \probe_\tindx^\prime (\obsresponse_\sindx -\obsresponse_\tindx)-\lambda_\tindx \Phi \leq 0 \quad  \\
& \lambda_t > 0 \quad \Phi \geq 0 \quad\text{for}\quad \tindx,\sindx\in \{1,2,\dots,\horizon\}.\label{eqn:AE}
\end{array}
\end{equation}
(iii) $f_M$ is the pdf of the random variable $M$ where
\begin{equation}
	M\triangleq\underset{\underset{\tindx \ne \sindx}{\tindx,\sindx}}{\max}\left[\probe_\tindx^\prime(\anoise_\tindx - \anoise_\sindx)\right]. 
	\label{eqn:def:M}
      \end{equation}
      The intuition behind   (\ref{eqn:Statistical_Test}), (\ref{eqn:AE}) is clear: if $\Phi = 0$, then  (\ref{eqn:AE}) is equivalent to  Afriat's theorem. Due to presence of noise, it is unlikely that $\Phi=0$ is feasible; so we seek the minimum perturbation $\Phi^*(\Obsresponse)$ that satisfies (\ref{eqn:AE}).

The constrained optimization problem (\ref{eqn:AE}) is non-convex due to the bilinear constraints $\lambda_\tindx \Phi$. However, since the
objective function depends only on the scalar $\Phi$, a one dimensional line search algorithm can be used. In particular, for any
fixed value of $\Phi$, (\ref{eqn:AE})  becomes a set of linear inequalities, and
so feasibility is straightforwardly determined.

The numerical implementation of  detector (\ref{eqn:Statistical_Test}) for a  given probe sequence $\Probe=[\probe_1,\ldots,\probe_\horizon]$ is described in Algorithm~\ref{alg:detect}.

\begin{algorithm}
\begin{compactenum}
\item  {\em Offline Step}.
For iterations $l=1,\ldots L$:
\begin{compactenum} \item 
  Simulate  noise sequence $\Anoise^{(l)}= [\anoise_1,\ldots,\anoise_\horizon]^{(l)}$.
\item Compute $M^{(l)}$  using (\ref{eqn:def:M}).
\end{compactenum}
 Compute the empirical distribution $\hat{F}_M(\cdot)$  of $M$ from these $L$ samples.
 \item Record the response $\Obsresponse$ from the radar to our probe $\Probe$.
   \item Solve  (\ref{eqn:AE}) for $\Phi$. Finally implement detector (\ref{eqn:Statistical_Test}) as
 \beq 1- \hat{F}_M(\Phi^*(\Obsresponse)) \overset{H_0}{\underset{H_1}{\gtrless}} \threshold  \label{eq:implement} \eeq
\end{compactenum}
\caption{Detecting Utility Optimizer given noisy response} \label{alg:detect}
\end{algorithm}

Note that Step 1 of Algorithm \ref{alg:detect} is offline; it  evaluates the empirical cdf $\hat{F}_M$. Step 2 records the noisy response  $\Obsresponse$ of the radar to our probe and finally Step 3 implements the detector with significance level $\gamma$.

The following theorem is our  main result for characterizing the detector  (\ref{eqn:Statistical_Test}). It states that the probability of Type-I error (false alarm) of the detector is bounded by $\threshold$ and that the optimal solution $ \Phi^*(\Obsresponse)$ gives the tightest false alarm bound.

      \begin{theorem} \label{thm:type1} Consider the noisy data set (\ref{eqn:noisydataset}) where
${\Obsresponse}=\left[\obsresponse_1,\obsresponse_2,\dots,\obsresponse_\horizon\right]$
        and detector (\ref{eqn:Statistical_Test}). 
        \begin{compactenum} \item  Suppose (\ref{eqn:AE}) has a feasible solution. Then
          $H_0$ is equivalent to the event that  $ \Phi^*(\Obsresponse) \leq M$ in (\ref{eqn:AE}).
          \item 
        The probability of Type-I error (false alarm) is
\beq P_{\Phi^*({\Obsresponse})}(H_1|H_0) \leq \threshold \label{eq:pfa} \eeq
\item   The optimizer $ \Phi^*({\Obsresponse})$ in   (\ref{eqn:Statistical_Test}) yields the tightest Type-I error bound, in that for any other  $\tPhi \in [\Phi^*, M]$,
      \beq  P_{\tPhi({\Obsresponse})}(H_1|H_0) \geq  P_{\Phi^*({\Obsresponse})}(H_1|H_0) . \label{eq:type1}
\eeq
      \end{compactenum}
    \end{theorem}

\begin{proof} Suppose $H_0$ holds. By. Theorem \ref{thm:AfriatTheorem}, $H_0$ is equivalent to    (\ref{eqn:AfriatFeasibilityTest}) having a feasible solution.
Let $(\lambdat_\tindx, \ut_\tindx)$ denote a feasible solution for  (\ref{eqn:AfriatFeasibilityTest}).
Then substituting $\response_\dtime = \nresponse_\dtime - \anoise_\dtime$, 
it is easily seen that $(\lambdat_\tindx, \ut_\tindx,\Phi = M)$  is   a feasible solution  for the noisy inequalities (\ref{eqn:AE}). Since  $(\lambdat_\tindx, \ut_\tindx,\Phi = M)$ is feasible, clearly 
  the minimizing solution  of   (\ref{eqn:AE}) satisfies $\Phi^*(\Obsresponse) \leq M$. 
  Therefore,   $$\text{    (\ref{eqn:AE})  feasible and  $H_0$  } \implies  \Phi^*(\Obsresponse) \leq M$$
  Similarly, let $(\bar{\lambda}_\tindx,\bar{u}_\tindx)$ denote a feasible solution to the noisy inequalities  (\ref{eqn:AE}). Then  $\Phi^*(\Obsresponse) \leq M$ implies that  (\ref{eqn:AfriatFeasibilityTest}) has a feasible solution,  i.e.,
  $$ \text{    (\ref{eqn:AE}) feasible  and }   \Phi^*(\Obsresponse) \leq M \implies H_0 $$ 
  Therefore if  (\ref{eqn:AE}) is feasible,  $H_0$ is equivalent to  $ \Phi^*(\Obsresponse) \leq M$.

  Let $\ccdf$ denote the complementary cdf of $M$.
From the statistical test  (\ref{eqn:Statistical_Test}), the event   $H_1$ given $H_0$  is equivalent to the event $\{\ccdf(\Phi^*(\Obsresponse)) \leq \threshold\}$ given
$ \{\Phi^*(\Obsresponse) \leq M\}$ and (\ref{eqn:AE}).
So $P(H_1|H_0) = P(\ccdf(\Phi^*(\Obsresponse)) \leq \threshold | \Phi^*(\Obsresponse) \leq M)$.
Now if $ \Phi^*(\Obsresponse) =M$, then since $\ccdf(M)$ is uniform\footnote{Obviously  the cdf and complementary cdf of any random variable $X$, namely, $\cdf(X)$ and $\ccdf(X)$,  are  uniformly distributed in $[0,1]$ } in $ [0,1]$ clearly
$P(H_1|H_0) = \threshold$. So if $  \Phi^*(\Obsresponse) \leq M$ then $P(\ccdf(\Phi^*(\Obsresponse)) \leq \threshold | \Phi^*(\Obsresponse) \leq M) \leq \threshold$, i.e., (\ref{eq:pfa}) holds.

Suppose $\tPhi(\Obsresponse) > \Phi^*(\Obsresponse)$. Then clearly,
$P(\ccdf(\tPhi(\Obsresponse)) \leq \threshold | \tPhi(\Obsresponse) \leq M) \geq P(\ccdf(\Phi^*(\Obsresponse)) \leq \threshold | \Phi^*(\Obsresponse) \leq M)$, i.e., (\ref{eq:type1}) holds.

\end{proof}

\subsection{Detecting Cognitive Radar given Noisy Probe}

Here we consider the case where the radar observes our probe signal $\probe_\dtime$ in additive noise $\pnoise_\dtime$ as
\begin{equation}
	\nprobe_\dtime = \probe_\dtime + \pnoise_\dtime.
	\label{eqn:probenoisemodel2}
\end{equation}
Here $\pnoise_\dtime$ are $\probedim$-dimensional i.i.d.\ random variables. Note that (\ref{eqn:probenoisemodel2})  is equivalent to the radar input being $\probe_\dtime$ and us observing the radar input in noise as $\nprobe_\dtime$.

Given the noisy data set 
\begin{equation}
	\obsdataset= \left\{\left(\nprobe_\dtime,\response_\dtime \right): \dtime \in \left\{1,\dots,\horizon \right\}\right\},
	\label{eqn:noisydataset2}
\end{equation}
we  propose the following statistical test for testing utility maximization~\eqref{eq:utilitymaximization} of the radar:
\begin{equation}
\boxed{\phantom{\text{\hspace{0.2mm}}}\int\limits_{\Phi^*(\Obsprobe)}^{+\infty}f_M(\psi)\mathrm{d}\psi \overset{H_0}{\underset{H_1}{\gtrless}} \threshold}. 
	\label{eqn:Statistical_Test2}
\end{equation}
In the statistical test (\ref{eqn:Statistical_Test2}): \\ (i) $\threshold$ is the ``significance level'' of the test. 
\\(ii) The test statistic $\Phi^*(\Obsprobe)$, with ${\Obsprobe}=\left[\nprobe_1,\nprobe_2,\dots,\nprobe_\horizon\right]$ is the solution of the following constrained optimization problem:
\begin{equation}
\begin{array}{rl}
\min & \Phi \\
\mbox{s.t.} & u_{\sindx}-u_{\tindx}- \lambda_\tindx \nprobe_\tindx^\prime (\response_\sindx -\response_\tindx)-\lambda_\tindx \Phi \leq 0 \quad  \\
& \lambda_t > 0 \quad \Phi \geq 0 \quad\text{for}\quad \tindx,\sindx\in \{1,2,\dots,\horizon\}.\label{eqn:AE2}
\end{array}
\end{equation}
(iii) $f_M$ is the pdf of the random variable $M$ where
\begin{equation}
	M\triangleq\underset{\underset{\tindx \ne \sindx}{\tindx,\sindx}}{\max}\left[\pnoise_\tindx^\prime(\response_\tindx - \response_\sindx)\right]. 
	\label{eqn:def:M2}
      \end{equation}
 The numerical implementation of the detector (\ref{eqn:Statistical_Test2}) for a given response $\Response = [\response_1,\ldots,\response_N]$ is  described in Algorithm \ref{alg:detect2}.

\begin{algorithm} 
\begin{compactenum}
\item  {\em Offline Step}.
For iterations $l=1,\ldots L$:
\begin{compactenum} \item 
  Simulate  noise sequence $\Pnoise^{(l)}= [\pnoise_1,\ldots,\pnoise_\horizon]^{(l)}$.
\item Compute $M^{(l)}$  using (\ref{eqn:def:M2}).
\end{compactenum}
 Compute the empirical distribution $\hat{F}_M(\cdot)$  of $M$ from these $L$ samples.
 \item Record the response $\Response$ from the radar to our noisy probe~$\bar{\Probe}$.
   \item Solve  (\ref{eqn:AE2}) for $\Phi$. Finally implement detector (\ref{eqn:Statistical_Test2}) as
 \beq 1- \hat{F}_M(\Phi^*(\bar{\Probe})) \overset{H_0}{\underset{H_1}{\gtrless}} \threshold  \label{eq:implement2} \eeq
\end{compactenum}
\caption{Detecting Utility Optimizer when Utility Optimizer views our probe signal  in noise} \label{alg:detect2} 
\end{algorithm}

In complete analogy to Theorem \ref{thm:type1} we have the following Type-I error bound for the detector in Algorithm \ref{alg:detect2}.

\begin{theorem}
  Consider the noisy data set (\ref{eqn:noisydataset2}) and statistical detector~ (\ref{eqn:Statistical_Test2}), (\ref{eqn:AE2}). Then the assertions of  Theorem \ref{thm:type1} hold.
  \end{theorem}

\subsection{Lower bound for False Alarm Probability}

Given the significance level of the statistical test in~\eqref{eqn:Statistical_Test}, a Monte Carlo simulation is required to compute the threshold.
We now present an analytical expression for a lower bound on the false alarm probability of the statistical test in~\eqref{eqn:Statistical_Test} when the additive noise $\anoise_\dtime$ in (\ref{eqn:noisemodel}) are standard normal variables.

\begin{theorem}[\cite{AK17}]
 Consider the noisy data set $\obsdataset$ in (\ref{eqn:noisydataset}). Suppose  $\{\anoise_\dtime\}$ in~\eqref{eqn:noisemodel} are i.i.d $N(0,I)$  Gaussian vectors and  ${\Obsresponse}=\left[\obsresponse_1,\obsresponse_2,\dots,\obsresponse_\horizon\right]$.
Then  the probability of false alarm in~\eqref{eqn:hypothesis} is lower bounded by 
	\begin{equation}
		1 - \underset{\dtime}{\prod} \left\{1 - \sqrt{\frac{2}{\pi}} \frac{\sqrt{2{\norm{\probe_\dtime}}^2}\exp{(-{\Phi^*(\Obsresponse )}^2/4{\norm{\probe_\dtime}}^2)}}{{\Phi^*(\Obsresponse)} + \sqrt{{\Phi^*(\Obsresponse )}^2+8{\norm{\probe_\dtime}}^2}}\right\}.
	  	\label{eqn:lowerbound}
	\end{equation}
	\label{thm:lowerbound}
\end{theorem}
The proof is in \cite{AK17}.
From the analytical expression  (\ref{eqn:lowerbound}),  we can obtain an upper bound of the test statistic, denoted by $\overline{\Phi^*(\bf \obsresponse)}$. 
Hence, given a data set $\obsdataset$ in~\eqref{eqn:noisydataset}, if the solution to the optimization problem~\eqref{eqn:AE} is such that $\Phi >  \overline{\Phi^*(\bf \obsresponse)}$, then the conclusion is that the data set does not satisfy utility maximization, for the desired false alarm probability. 

{\em Remark}. We have discussed detecting a cognitive radar when either the radar's response is observed in noise or our probe vector to the radar is observed in noise. A more general framework would be where both probe and response were observed in noise. We are unable to analyze the detector in this case.

\section{Adaptive Optimization of   Probe Signal  to Minimize  Type-II Detection Error Probability} \label{sec:adapt}
This section deals with adaptively interrogating the enemy radar to detect if it is cognitive, based on noisy measurements of the radar's response.  Specifically,
given batches of noisy measurements of the enemy's radar  response $\Obsresponse_k=[\obsresponse_{1,k},\ldots,\obsresponse_{\horizon,k}]$, $k=1,2,\ldots$ (see (\ref{eqn:noisemodel})) how can we adaptively design batches of our probe signals   $\Probe_k = [\probe_{1,k},\ldots,\probe_{\horizon,k}]$, $k=1,\ldots,$  so as to   minimize  the Type-II error (deciding that the radar is cognitive when it is not)?

Theorem~\ref{thm:type1}  above guarantees that if we observe the radar response in noise, then the probability of Type-I errors (deciding that the radar is not cognitive when it is) is less then $\threshold$ for the decision test (\ref{eqn:Statistical_Test}).  Our aim is to enhance the statistical test (\ref{eqn:Statistical_Test})  by adaptively choosing the probe vectors $\Probe=[\probe_1,\probe_2,\dots,\probe_\horizon]$ to reduce the probability of  Type-II errors
(deciding that the radar is cognitive when it is not).

The framework is shown in Figure \ref{fig:optimize} and can be viewed as a form of active inverse reinforcement learning.

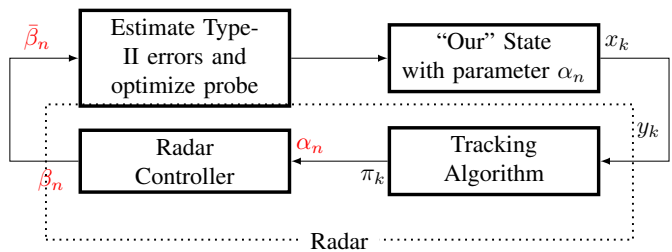
\begin{figure}[h]
\resizebox{9cm}{!}{ 
\begin{tikzpicture}[node distance = 4.5cm, auto]
  \node [block] (BLOCK1) {Estimate Type-II errors and optimize probe};
  \node[block,right of = BLOCK1](opt){``Our'' State \\ with parameter $\probe_\dtime$ };
    \node [block, below of=BLOCK1,node distance=1.5cm] (BLOCK2) {Radar\\ Controller};
    \node[block, right of=BLOCK2] (grad) {Tracking\\ Algorithm};

    \path [line] (BLOCK1) -- (opt);
    
    \path [line] (opt.east) --++ (1cm,0cm) node[pos=0.25,above]{$\state_k$}  |-   node[pos=0.35,left]{$\obs_k$}    (grad);
    
     \path [line] (BLOCK2.west) --++ (-1cm,0cm)  node[red, pos=0.4,below]{$\response_\dtime$}  |-     node[red, pos=0.7,above]{$\obsresponse_\dtime$}     (BLOCK1.west);
     \path [line] (grad) --   node[pos=0.15,below]{$\belief_k$}    node[red, pos=0.8,above]{$\probe_\dtime$}  (BLOCK2);
     \node[draw, thick, dotted, inner sep=3ex, yshift=-1ex,
     fit=(BLOCK2) (grad)] (box) {};
     \node[fill=white, inner xsep=1ex] at (box.south) {Radar};
   \end{tikzpicture}
 }
 \caption{Optimizing the probe waveform to detect cognition in adversary's radar by minimizing the Type-II errors subject to constraints in Type-I errors.
   }
 \label{fig:optimize}
\end{figure}

The probe signals $\Probe$ are adapted to estimate
{\normalsize 
\begin{align}
&\argmin_{\Probe\in\reals^{\probedim\times \horizon}_+}J(\Probe) = \underbrace{\prob\big(\!\!\!\!\!\!\int\limits_{\Phi^*\big(\Response(\Probe)+\Anoise\big)}^{+\infty} \!\!\!\!\!\! f_{M}(\psi) d\psi > \threshold \big| \{\Probe,\Response(\Probe)\}\in \mathcal{A}\big)}_\text{Probability of Type-II error}. 
\label{eqn: SPSA Objective}
\end{align}
}
Here $\prob(\cdots | \cdot)$ denotes the conditional probability that the statistical test (\ref{eqn:Statistical_Test}) accepts $H_0$, defined in (\ref{eqn:hypothesis}), given that $H_0$ is false. In (\ref{eqn: SPSA Objective}), the noise matrix $\Anoise=[\anoise_1,\anoise_2,\dots,\anoise_\horizon]$ where
the random vectors $\anoise_\dtime$ are defined in
(\ref{eqn:noisemodel}), and $\threshold$ is the significance level of  (\ref{eqn:Statistical_Test}).
 The set $\mathcal{A}$ contains all the elements $\{\Probe,\Response(\Probe)\}$, with $\Response(\Probe)=[\response_1,\response_2,\dots,\response_\horizon]$, where  $\{\Probe,\Response\}$ does not satisfy~(\ref{eq:garp}).

Since the probability density function $f_M$ defined in (\ref{eqn:def:M}) is not known explicitly, (\ref{eqn: SPSA Objective}) is a simulation based stochastic optimization problem. To determine a local minimum value of the Type-II error probability $J(\Probe)$ wrt $\Probe$, several types of stochastic optimization algorithms can be used~\cite{Spa03}. Algorithm \ref{alg:spsa}  uses the simultaneous perturbation stochastic gradient (SPSA) algorithm:

\begin{algorithm}
\begin{compactenum}
\item[\bf Step 1.]  Choose initial probe $\Probe_0=[\probe_1,\probe_2,\dots,\probe_\horizon]\in\reals^{m\times n}_+$
  \item[\bf Step 2.]  For iterations $k=1,2,3,\dots$
\begin{compactenum}
\item Estimate empirical Type-II error probability  $J(\Probe_k)$ in (\ref{eqn: SPSA Objective}) using  $S$ independent trials: 
\begin{equation}
\hat{J}_k(\Probe_k) =\frac{1}{S}\sum\limits_{s=1}^{S}\mathbf{1}{\biggl( \hat{F}_M(\Phi^*(\Obsresponse_s)) \leq 1-\threshold\biggr)}
\label{eqn:SPSACost}
\end{equation}
Here, in each trial   $\Obsresponse_s$ is the noisy measurement of the radar response  to our probe vector~$\Probe_k$, see  (\ref{eqn:noisemodel}).
Also $\mathbf{1}(\cdot)$ is the indicator function.  $\hat{F}_M(\Phi^*(\Obsresponse_s))$ is the empirical cdf of $M$ computed  as in (\ref{eq:implement}).
 $\Phi^*(\Obsresponse_s)$ is obtained from (\ref{eqn:AE})  using  noisy observation sequence $\Obsresponse_s$ where $\Anoise_s$ is  a fixed realization of $\Anoise$, and  data set $\{\Probe_k,\Response(\Probe_k)\}\in\mathcal{A}$, described below (\ref{eqn: SPSA Objective}).

\item  Compute the gradient estimate $\hat{\nabla}_{\Probe}$
\begin{align}
\hat{\nabla}_{\Probe}\hat{J}_k(\Probe_k) &=  \frac{\hat{J}_k(\Probe_k+\Delta_k\omega)-\hat{J}_k(\Probe_k-\Delta_k\omega)}{2\omega\Delta_k} \label{eqn:SPSA}\\
\Delta_k(i) &= \begin{cases}
   -1 & \text{with probability 0.5} \\
   +1       & \text{with probability 0.5} 
  \end{cases} \nonumber
\end{align}
with gradient step size $\omega > 0$. 
\item  Update the probe vector $\Probe_k$ with step size $\mu>0$:
\begin{equation*}
\Probe_{k+1} = \Probe_k-\mu\hat{\nabla}_{\Probe}\hat{J}_k(\Probe_k).
\end{equation*}
\end{compactenum}
\end{compactenum}
\caption{SPSA for Minimizing Type-II Error Probability}
\label{alg:spsa}
\end{algorithm}

For the reader who is familiar with adaptive filtering algorithms, the SPSA can be viewed  as a generalization where an explicit formula
for the gradient is not available and needs to be estimated
by stochastic simulation.
 A useful property of the SPSA algorithm is that estimating the gradient $\nabla_{\Probe} J_k(\Probe_k)$ in (\ref{eqn:SPSA})  requires only two measurements of the cost function (\ref{eqn:SPSACost}) corrupted by noise per iteration, i.e., the number of evaluations is independent of the dimension $\probedim\times \horizon$ of the vector $\Probe$. In comparison, a naive finite difference gradient estimator requires computing $2(\probe \times \horizon)$ estimates of the cost  per iteration; see \cite{Spa03} for a tutorial exposition of the SPSA algorithm. For decreasing step size $\mu = 1/k$, the SPSA algorithm converges with probability one to a local stationary point. For constant step size $\mu$, it converges weakly  \cite{KY03}.  

 \section{Numerical Examples}

 This section presents four classes of numerical examples to illustrate the key results of the paper.
 
 \subsection{Spectral Revealed Preferences with Linear Budget}
 \label{sec:linearsim}

We illustrate  identification  of a  cognitive radar that optimizes its waveform
subject to  linear budget constraint (\ref{eq:radaropt}); see Sec.\ref{sec:linear} for detailed motivation. For easy visualization, we chose  $\probedim=2$ (dimension of probe signal) so that the estimated utility function can be displayed on a 2-d contour plot.

The elements of our probe signal $\probe_\dtime$ are generated  randomly and  independently over time $\dtime$ as
$\probe_\dtime(1) \sim \uniform(0.1, 1.1)$ and $\probe_\dtime(2) \sim \uniform(0.1,1.1)$ where $\uniform(a,b)$ denotes uniform pdf with support $(a,b)$. Recall  our probe signal specifies the diagonal state covariance matrix   $\snoisecov_\dtime = \diag[\probe_\dtime(1), \probe_\dtime(2)]$ in~(\ref{eq:lineargaussian}).

In response to 
 our probe  vector sequence $\{\probe_\dtime,\dtime = 1,\ldots, 50\}$,
suppose the radar chooses its waveform parameters  (e.g.,  triangular waveform parameters in  (\ref{eq:triangular}) or
Gaussian pulse parameters in (\ref{eq:gaussianpulse}))  as
$$ \response_\dtime = \argmax_\response 
\utility(\response) = \det(\onoisecov^{-1}(\response)) = \response(1) \times \response(2)$$  subject to  linear budget constraint (\ref{eq:radaropt}), namely
$\probe_n^\p \response \leq 1$.

Given the dataset $ \dataset=\{(\probe_\dtime,\response_\dtime), \dtime\in \{1,2,\ldots,50\}\}$, how to detect if the radar is a constrained utility maximizer (cognitive)? We verified that Afriat's inequalities (\ref{eqn:AfriatFeasibilityTest}) have a feasible solution  implying that the radar's response is consistent with  utility maximization.  Also the set of utility functions consistent with $\dataset$ can be reconstructed via (\ref{eqn:estutility}).
Figure~\ref{fig:det} shows the contours of one such  utility function that rationalizes the dataset $\dataset$.

It is instructive to compare the response of the radar when it maximizes other utility functions instead of the determinant. We chose $\utility(\response) =\trace(\onoisecov^{-1}(\response))$
For the same probe inputs as above, and the corresponding radar response,  we verified that Afriat's inequalities have a feasible solution, implying that the radar is a utility maximizer. 
Figure~\ref{fig:trace}
shows the contours of one such estimated utility function which rationalizes the radar's input output dataset.  As expected, the radar's response is aligned with the axes since it chooses  as much as possible of the “cheapest” option. Also the recovered utility is linear (the contours are lines).
Finally, we choose the Cobb-Douglas utility $\utility(\response) = \sqrt{\response(1)}\,\response(2)$. Figure~\ref{fig:cobb1} shows  the contours of one such estimated utility function which rationalizes the radar's input output dataset.

{\em Remarks}. (i) Note that the utility $\utility(\response) = \det(\onoisecov^{-1}(\response)) = \response(1) \times \response(2)$ is not a concave function of $(\response(1),\response(2))$. It is important to point out that  Afriat's Theorem \ref{thm:AfriatTheorem}  makes no assumption on concavity of the utility. Yet Afriat's theorem guarantees that the reconstructed utility function which  rationalizes the data is concave; see also Footnote~\ref{foot:varian} in Sec.\ref{subsec:afriat}. \\
(ii) We also verified that if the radar response $\response_\dtime$ is chosen as an independent random
sequence, then as might be expected, the radar does not satisfy the utility maximization test with high probability. \\

\begin{figure}[t]
  \centering
     \begin{subfigure}[b]{0.4\textwidth}
     \includegraphics[width=0.85\textwidth]{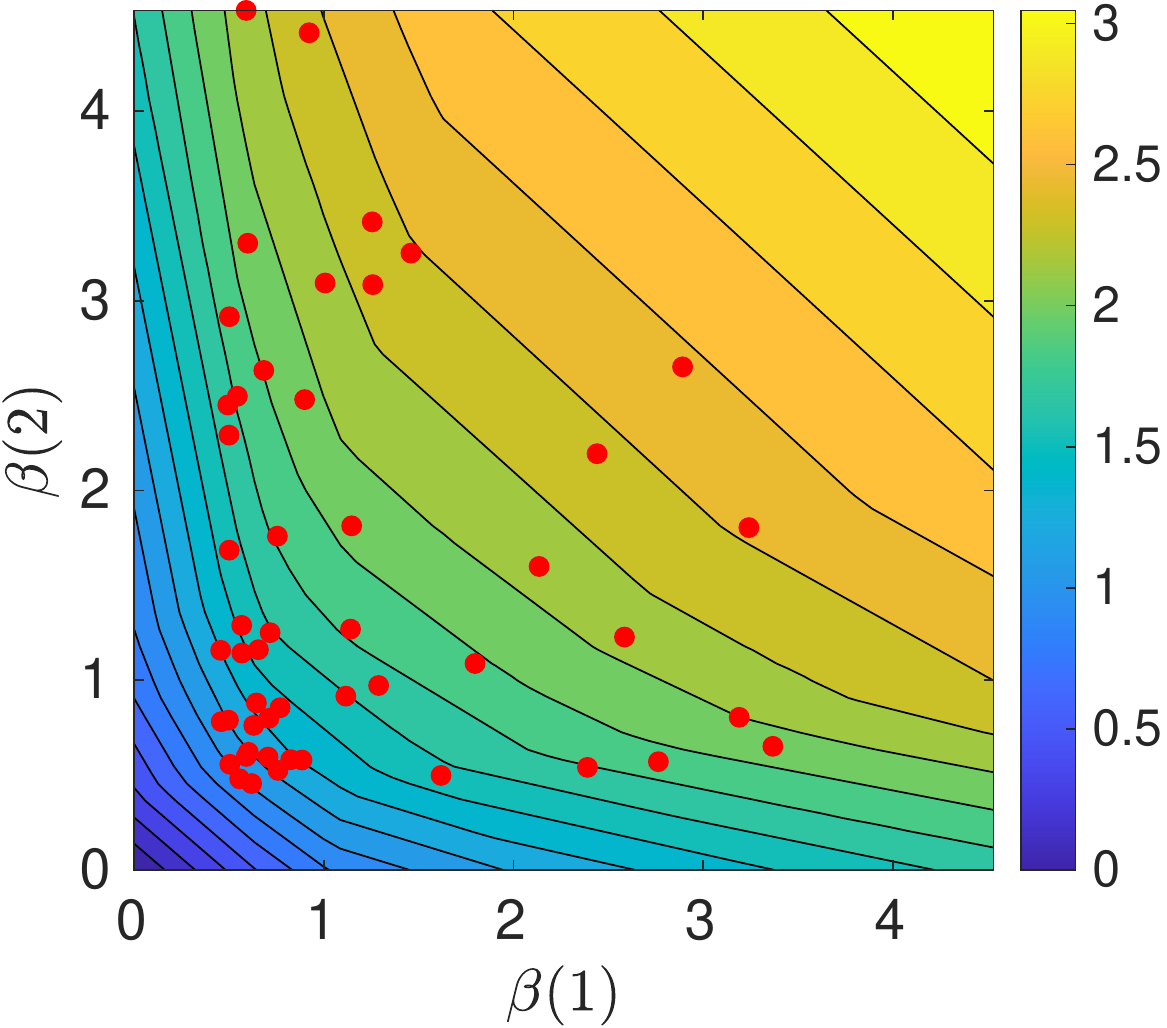}
     \caption{$\utility(\response) =  \det(\onoisecov^{-1}(\response))$.}
     \label{fig:det}
   \end{subfigure}
   \\ \smallskip
    \begin{subfigure}[b]{0.4\textwidth}
   \includegraphics[width=0.85\textwidth]{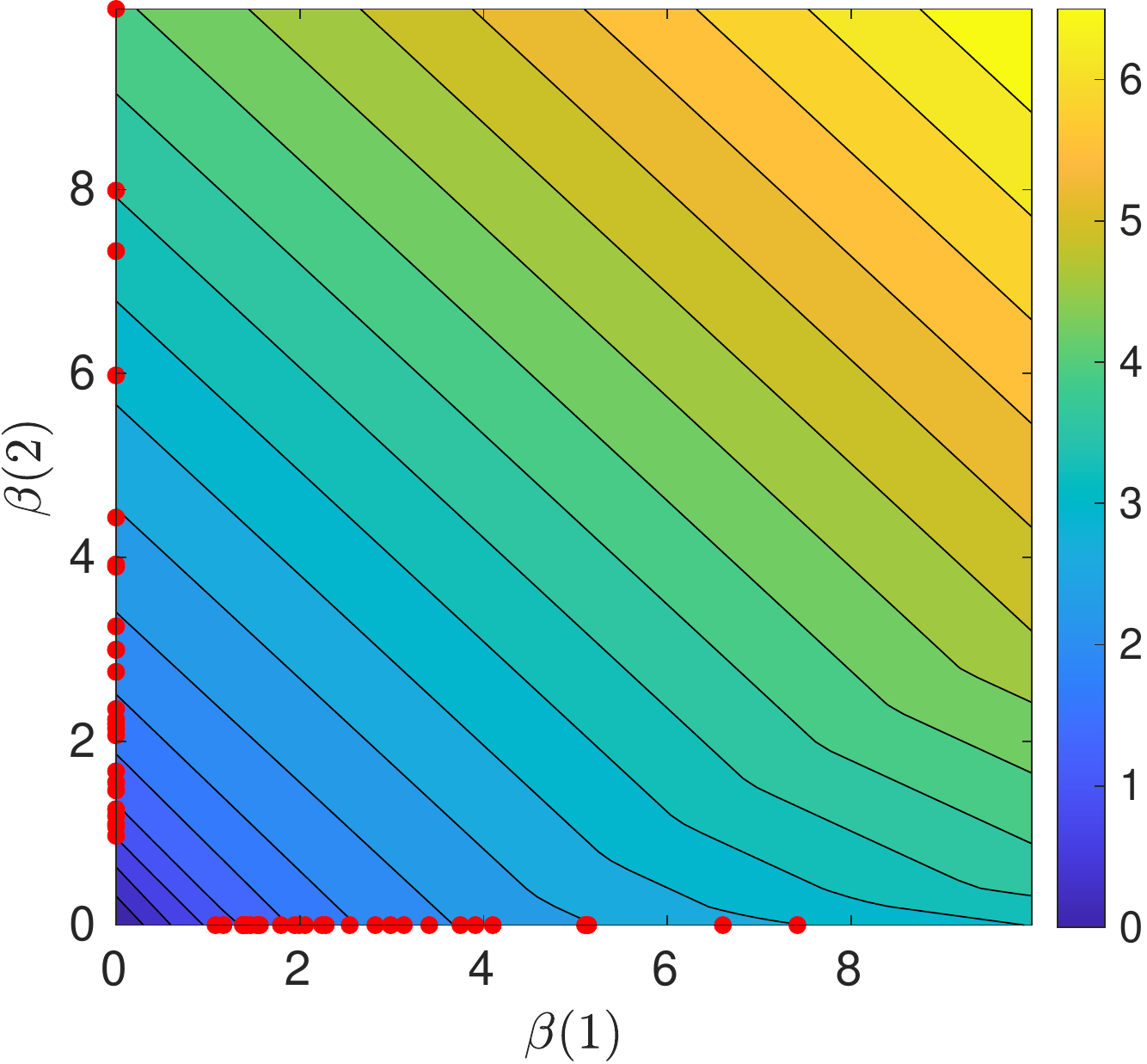}
   \caption{$\utility(\response) = \trace(\onoisecov^{-1}(\response)$.}
   \label{fig:trace}
  \end{subfigure}
   \\ \smallskip
    \begin{subfigure}[b]{0.4\textwidth}
   \includegraphics[width=0.85\textwidth]{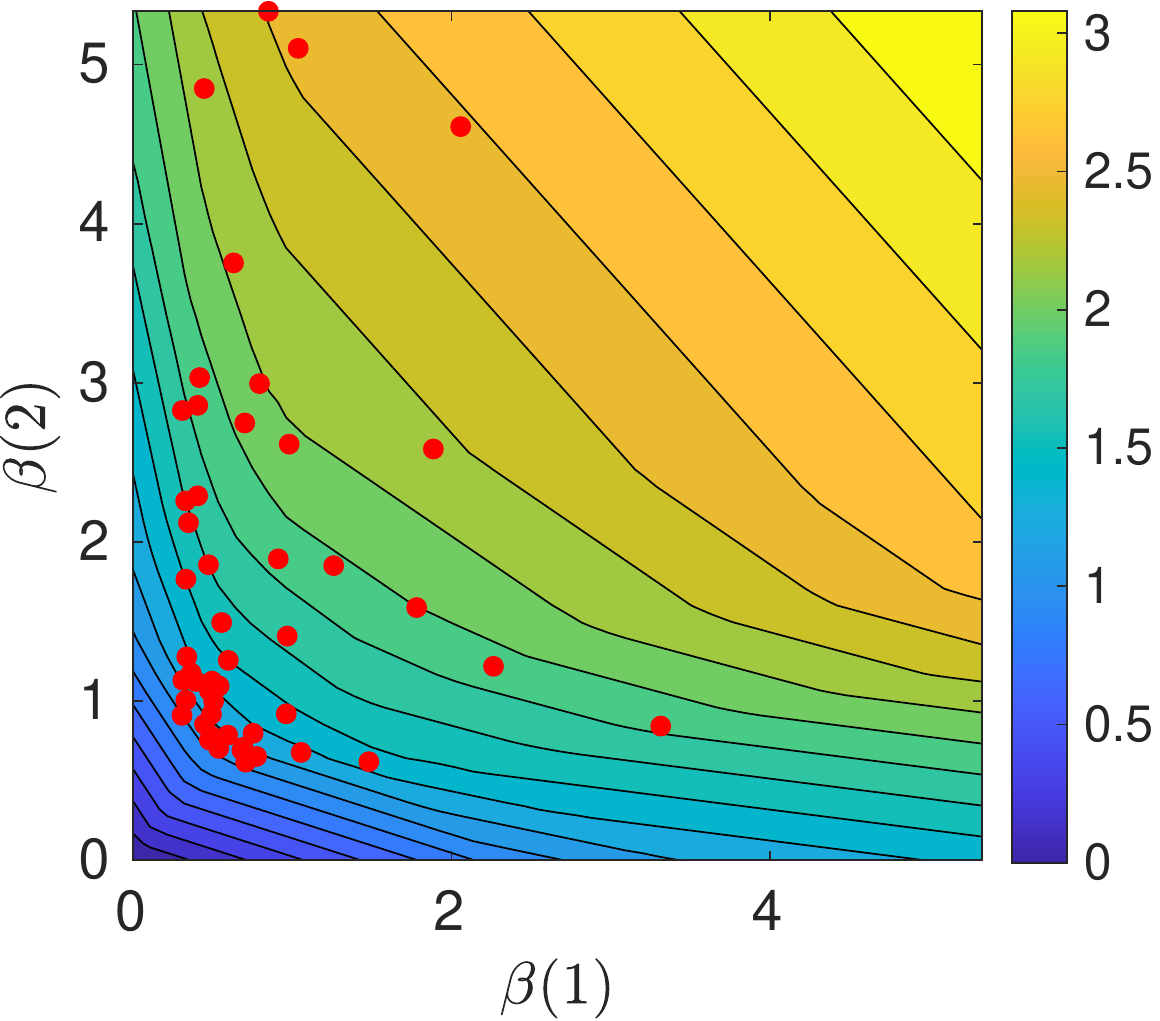}
   \caption{Cobb-Douglas utility  $\utility(\response) = \sqrt{\response(1)} \response
     (2) $.}
   \label{fig:cobb1}
  \end{subfigure}
   \caption{Recovered utility function for a cognitive radar which maximizes $\utility(\response)$ with  linear budget constraint $\probe_n^\p \response \leq 1$. The red dots show the response vectors $\response_\dtime$ selected by the radar.}
  \label{fig:simlinearbudget}
\end{figure}

\subsection{Spectral Revealed Preferences with Nonlinear Budget}
We now  illustrate  identifying a  cognitive radar that optimizes its waveform
subject to  nonlinear budget constraint (\ref{eq:radarnonlinear}); see Sec.\ref{sec:nonlinearbudget} for detailed motivation. The probe vectors were generated  as in Sec.\ref{sec:linearsim}, but recall from Sec.\ref{sec:nonlinearbudget} that now 
$\snoisecov^{-1}_\dtime = \mathrm{diag}(\probe_\dtime)$).  In response to probe $\probe_\dtime$,  the radar
chooses its waveform by choosing $\response_\dtime$ that maximizes   $\det(\onoisecov(\response)) = \response(1) \times \response(2)$ subject to  nonlinear budget constraint  (\ref{eq:radarnonlinear}).

The physical parameters for the state space model (\ref{eq:lineargaussian}) used by the radar's tracker and ARE (\ref{eq:are}): $A = \begin{bmatrix} 1 & 1 \\ 0 & 1 \end{bmatrix}$, $C = I_{2\times 2}$.
We chose the nonlinear budget constraint parameters that specify (\ref{eq:radarnonlinear})  as
$\bar{\response}_\dtime = \begin{bmatrix} 10 & 10 \end{bmatrix}'$ and $\bar{\lambda} = 3.6$.

Given the dataset $ \dataset=\{(\probe_\dtime,\response_\dtime), \dtime\in \{1,2,\ldots,50\}\}$, we verified that the linear  inequalities (\ref{eq:nonlinearFeasibilityTest}) of Theorem \ref{thm:nonlinear} 
 have a feasible solution. Therefore,  the radar's response is consistent with  utility maximization with a nonlinear budget.  Also the set of utility functions consistent with dataset $\dataset$ were  reconstructed via (\ref{eq:nonlinearutility}).
Figure~\ref{fig:simnonlinearbudget} shows the contours of one such  utility function that rationalizes the dataset $\dataset$.

\begin{figure}[h]
  \centering
  \includegraphics[width=0.35\textwidth]{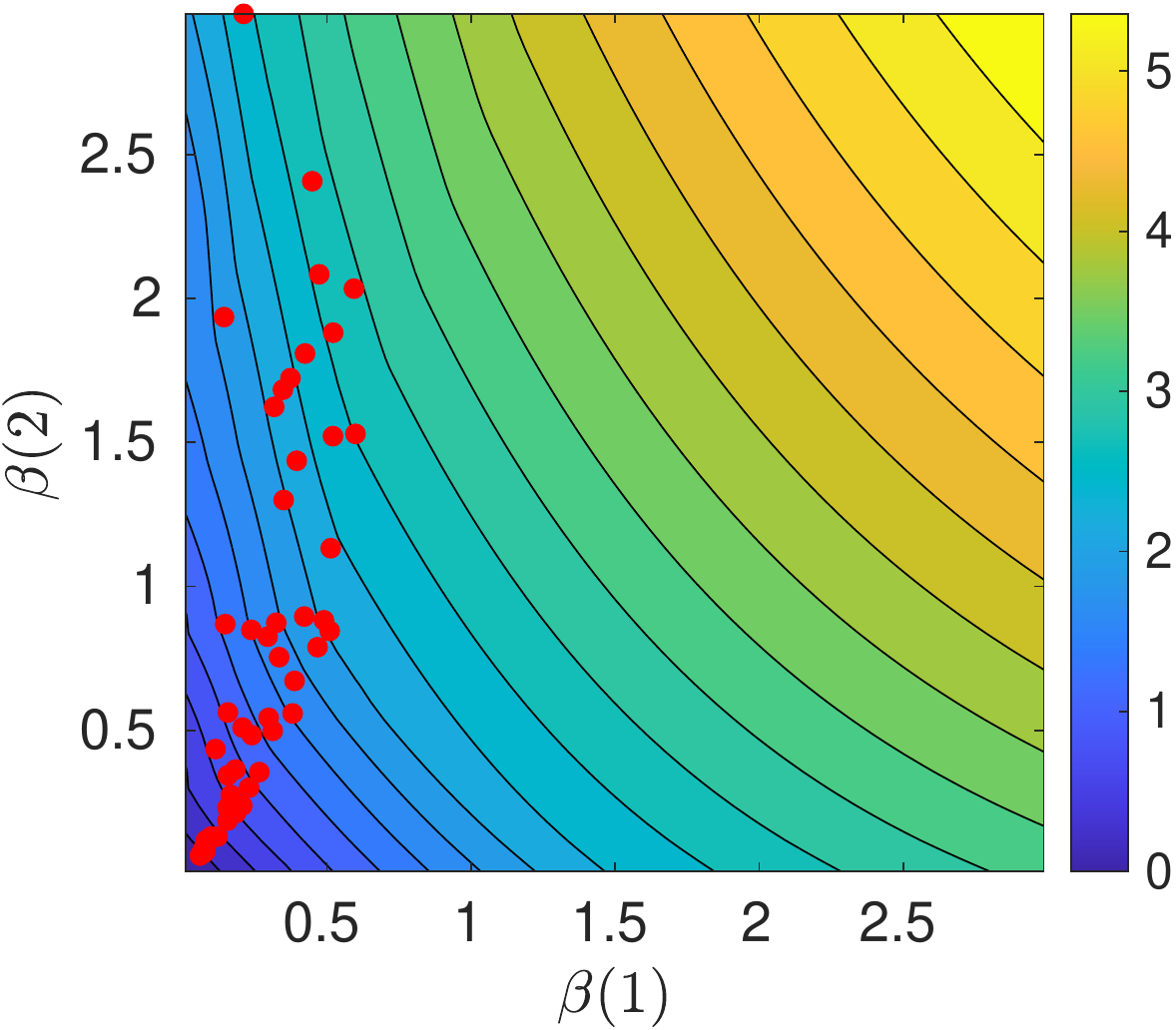}
  \caption{Recovered utility functions for a cognitive radar  which maximizes $\det(\onoisecov(\response))$ with a  nonlinear budget constraint. The red dots show the response vectors $\response_\dtime$ selected by the radar.}
  \label{fig:simnonlinearbudget}
\end{figure}

\subsection{Beam Allocation. Detecting Cognitive Radar in Noise}
\label{sec:beamsim}
 Here we  illustrate the statistical detectors Algorithm \ref{alg:detect} and Algorithm \ref{alg:detect2} described in Sec.\ref{sec:noise} where the response of the radar is observed in noise and the radar observes our probe signal in noise.
We consider the setup of Sec.\ref{sec:beam}  where a radar switches its beam between $\probedim=3$  targets over $\horizon=20$ epochs.
Recall that the $i$-th component of our probe signal, namely, $\probe_\dtime(i)$ is the  trace of the inverse predicted covariance matrix for target $i$.
We generated the probe signal as $\alpha_\dtime(i) \sim \uniform(0,0.05)$.  The response $\response_\dtime$ of the cognitive radar    which optimizes its beam allocation was obtained by maximizing the Cobb-Douglas utility\footnote{Cobb-Douglas utility functions are used widely in macroeconomics and resource allocation  \cite{Fra62,Gol68}.}
$$ \utility(\response) = \sqrt{\response(1) } \, \response(2) \, (\response(3))^2 $$
subject to the linear constraint specified in (\ref{eq:beam}).

For a non-cognitive radar, we simulated its response as the random outcome  $\response_\dtime(i) \sim \uniform(0,1)$ where the samples $\response(1) + \response(2) + \response(3) >1$ are discarded. Recall that $\response_\dtime$ is the vector comprising the fraction of time the radar allocates to each of the
$\probedim$ targets in epoch~$\dtime$.

\subsubsection{Detecting Cognitive Radar given noisy response}
Our noisy measurement $\nresponse_\dtime$ of the radar response $\response_\dtime$ was simulated as (\ref{eqn:noisemodel}) where the observation noise  $\anoise_\dtime(i) \sim \normal(0,\sigma^2)$, $i=1,\ldots,\probedim$,  over a range of values for $\sigma$.

Given  noisy data set $	\obsdataset= \left\{\left(\probe_\dtime,\obsresponse_\dtime \right): \dtime \in \left\{1,\dots,\horizon \right\}\right\}$ from the enemy radar,
Figure \ref{fig:noisyresponse} displays the statistic $ 1- \hat{F}_M(\Phi^*(\Obsresponse)) $ of our detector, namely left hand side of  (\ref{eq:implement}), for both cognitive and non-cognitive radar cases. Recall from (\ref{eq:implement})  that when this statistic exceeds  significance level $\threshold$, the radar is classified as cognitive.  As might be expected,  Figure \ref{fig:noisyprobe}  shows that
for low noise variance $\sigma^2$, the cognitive and non-cognitive cases are easily distinguished. But as $\sigma^2$ becomes larger, the non-cognitive radar can be mistakenly identified as cognitive, i.e., the probability of a Type-II error increases.

\begin{figure} \centering
     \begin{subfigure}[b]{0.48\textwidth}
       \includegraphics[width=0.85\textwidth]{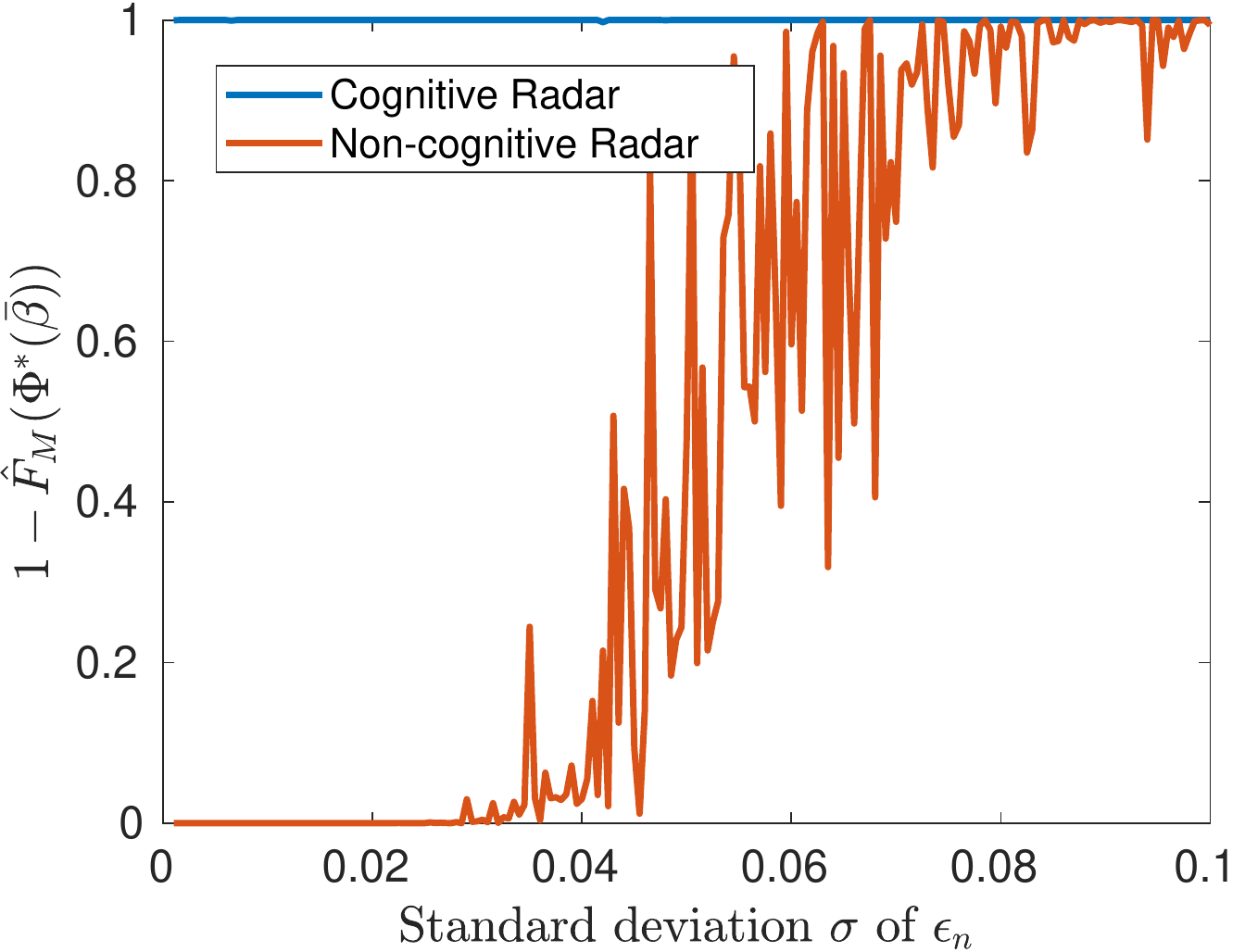}
        \caption{Statistic  $ 1- \hat{F}_M(\Phi^*(\Obsresponse))$ defined in  (\ref{eq:implement})  for cognitive and non-cognitive radar  over a range of response noise variance~$\sigma^2$. } \label{fig:noisyresponse}
\end{subfigure}
\\ \smallskip 
     \begin{subfigure}[b]{0.48\textwidth}
       \includegraphics[width=0.85\textwidth]{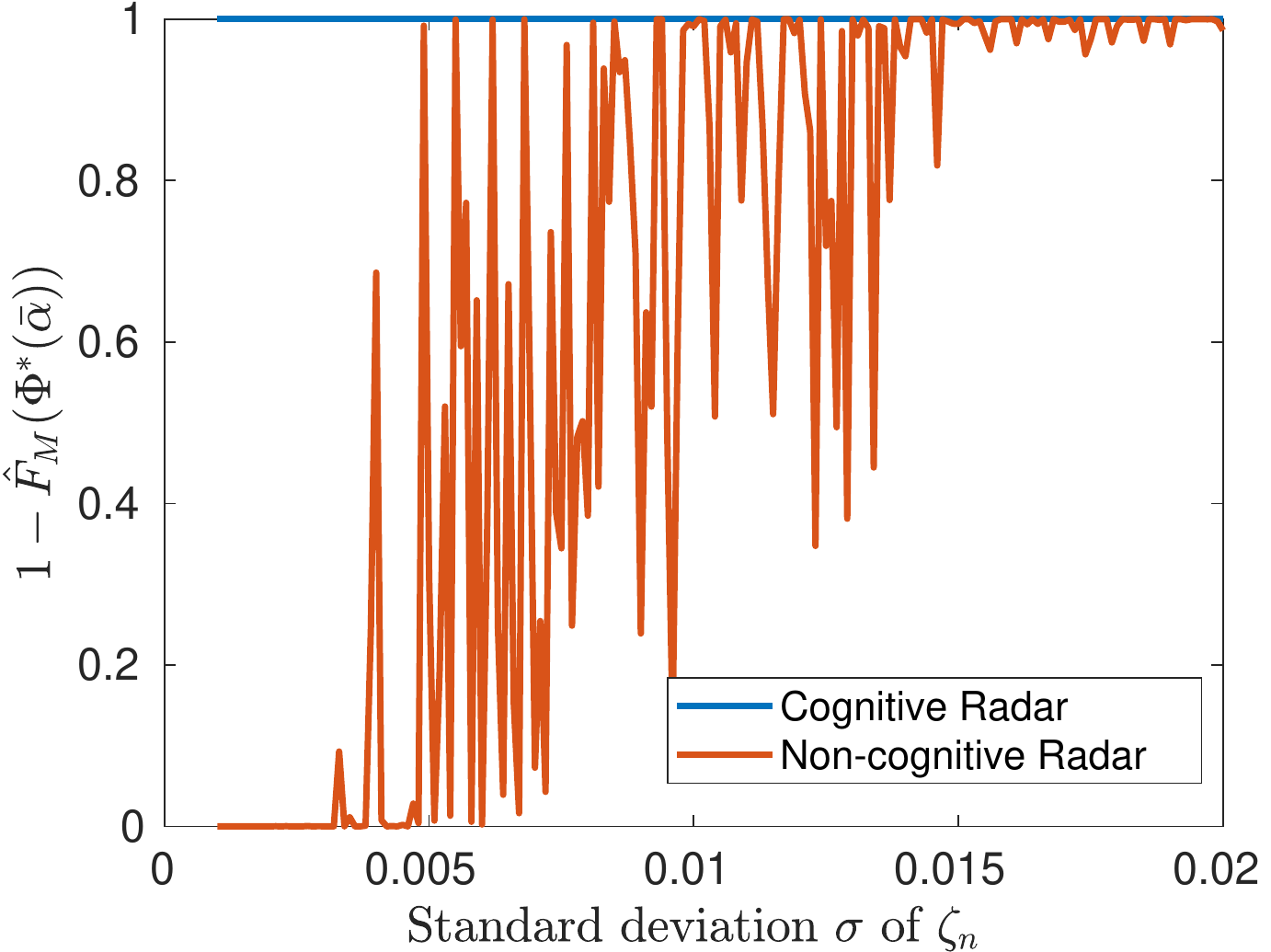}
          \caption{Statistic  $ 1- \hat{F}_M(\Phi^*(\Obsprobe))$  defined in  (\ref{eq:implement2}) for cognitive and non-cognitive radar  over a range of probe noise variance~$\sigma^2$. } \label{fig:noisyprobe}
        \end{subfigure}
        \caption{Statistical Detector for Cognitive Radar given Noisy measurements of Response and Probe. When either the red or blue curve is above (below) a value $\gamma$ (significance level), the detector classifies the radar as cognitive (resp.\ not cognitive).}
   \end{figure}

\subsubsection{Detecting Cognitive Radar given noisy probe}
Suppose the radar observes our probe signal $\probe_\dtime $  in noise as $\nprobe_\dtime$ as specified by (\ref{eqn:probenoisemodel2}) where the observation noise $\pnoise_\dtime(i) \sim \normal(0,\sigma^2)$ over a range of $\sigma$.
Given the noisy data set $	\obsdataset= \left\{\left(\nprobe_\dtime,\response_\dtime \right): \dtime \in \left\{1,\dots,\horizon \right\}\right\}$ from the enemy radar,
Figure \ref{fig:noisyprobe} displays the statistic  
$ 1- \hat{F}_M(\Phi^*(\Obsprobe)) $ of our detector, defined in  (\ref{eq:implement2}), for both cognitive and non-cognitive radar. As expected, Figure \ref{fig:noisyprobe} shows that increasing the noise variance $\sigma^2$ results in increasing the probability of Type-II error.

To summarize, Figures \ref{fig:noisyresponse} and \ref{fig:noisyprobe} display the performance of the statistical tests (Algorithms \ref{alg:detect} and \ref{alg:detect2}) for detecting cognitive radars. The figures show that when the probe or response are measured with small noise variance, it is relatively easy to distinguish between a cognitive and non-cognitive radar. But for  large noise variance, it becomes increasingly  difficult to distinguish between a cognitive and non-cognitive radar. Below we show that by optimizing our  probe signal, we can significantly improve the performance of the detector in high noise variance.

\subsection{Adaptive Optimization of Probe Signal to minimize Type-II detection error} \label{sec:adaptopt}

 This section  illustrates the framework of Sec.\ref{sec:adapt}. We observe the radar response in noise and  probe the radar to identify  if it is cognitive.  Using a numerical example, we show that by  adaptively optimizing our probe signal, we can substantially reduce the   Type-II error probability (identifying that the radar  is cognitive when it is not) while constraining the Type-I error. The adaptive optimization of the probe signal is carried out via the SPSA algorithm~\ref{alg:spsa} on the objective function (\ref{eqn: SPSA Objective}), namely the Type-II error probability.

The setup involves the radar beam scheduling discussed  in Sec.\ref{sec:beam}  with
$m = 3$ targets and
$N = 20$ epochs. We initialize the probe vectors as
$\probe_1(i) \sim \uniform(0, 0.05)$. Recall from (\ref{eq:probe_beam}) that the elements 
 $\probe_\dtime(i)$ of the probe signal  are the trace of predicted precision of target $i$.
We observe the response $\response_\dtime(i)$ of the enemy radar in additive noise as in (\ref{eqn:noisemodel}) where
$\anoise_\dtime(i) \sim \normal(0,\sigma^2)$ and  $\sigma = 0.1$. Recall  $\response_\dtime(i)$ is the fraction of time in epoch $\dtime$ that the enemy radar allocates to  target $i$.

For the cognitive radar, our simulation uses the Cobb-Douglas utility
$U(\response) =\response(1)^{0.2} \,\response(2)^{0.3}\,\response(3)^{0.5}$ for  the cognitive radar.
For the non-cognitive radar,  the utility chosen is
$U(\response) = \response(1)^{\zeta_1 } \response(2)^{\zeta_2}   \response(3) ^{\zeta_3}$ where the preferences  $\zeta_i$  are generated randomly with  $\uniform (0, 1)$ density at each epoch $k$ and then normalizing  the sum to make them add to 1. (So  the response of the non-cognitive radar are random iid variables).

With the above setup, we  used batches of 
$L = 1000$  samples  to estimate the empirical distribution of  $\hat{F}_M $ in (\ref{eq:implement}).
Then Algorithm \ref{alg:spsa} was run for 200 iterations with the following parameters:
$S = 100$  trials were used to evaluate the empirical Type-II error probability $\hat{J}_k(\Probe_k)$ in  (\ref{eqn:SPSACost}) with significance level $\gamma = 0.05$. The gradient  step size
$\omega = 0.005$ in (\ref{eqn:SPSA}), $\mu =  0.005/k$ for the SPSA step size in Step 2c of
Algorithm \ref{alg:spsa}. Figure \ref{fig:spsa} displays the performance of the SPSA algorithm. As can be seen, the Type-II error probability is decreased significantly (almost 80\%) by careful choice of the probe signal. Thus our statistical detector  (\ref{eqn:Statistical_Test})  can  adequately reject non-cognitive radars.

{\em Interpretation. The probe signal matters}: In Sec.\ref{sec:beamsim} we chose the probe signal as $\probe_\dtime \sim \uniform(0.005)$. Then for
$\sigma=0.1$, Figure \ref{fig:noisyresponse} shows  that  the test statistic for the cognitive and non-cognitive radar are  almost  identical; so  it is impossible to distinguish between a cognitive and non-cognitive radar. Yet by optimizing our input probe,  Figure \ref{fig:spsa} shows that we can reduce the Type-II error probability to less than 0.2. This is also apparent from
 Figure \ref{fig:spsa} where the empirical Type-II error probability starts at 1 in the initial iterations  and goes down to 0.2 after optimizing the probe signal. We conclude  that judicious choice of the probe signal is crucial in identifying a cognitive radar given noisy measurements. 

\begin{figure} \centering
  \includegraphics[scale=0.5]{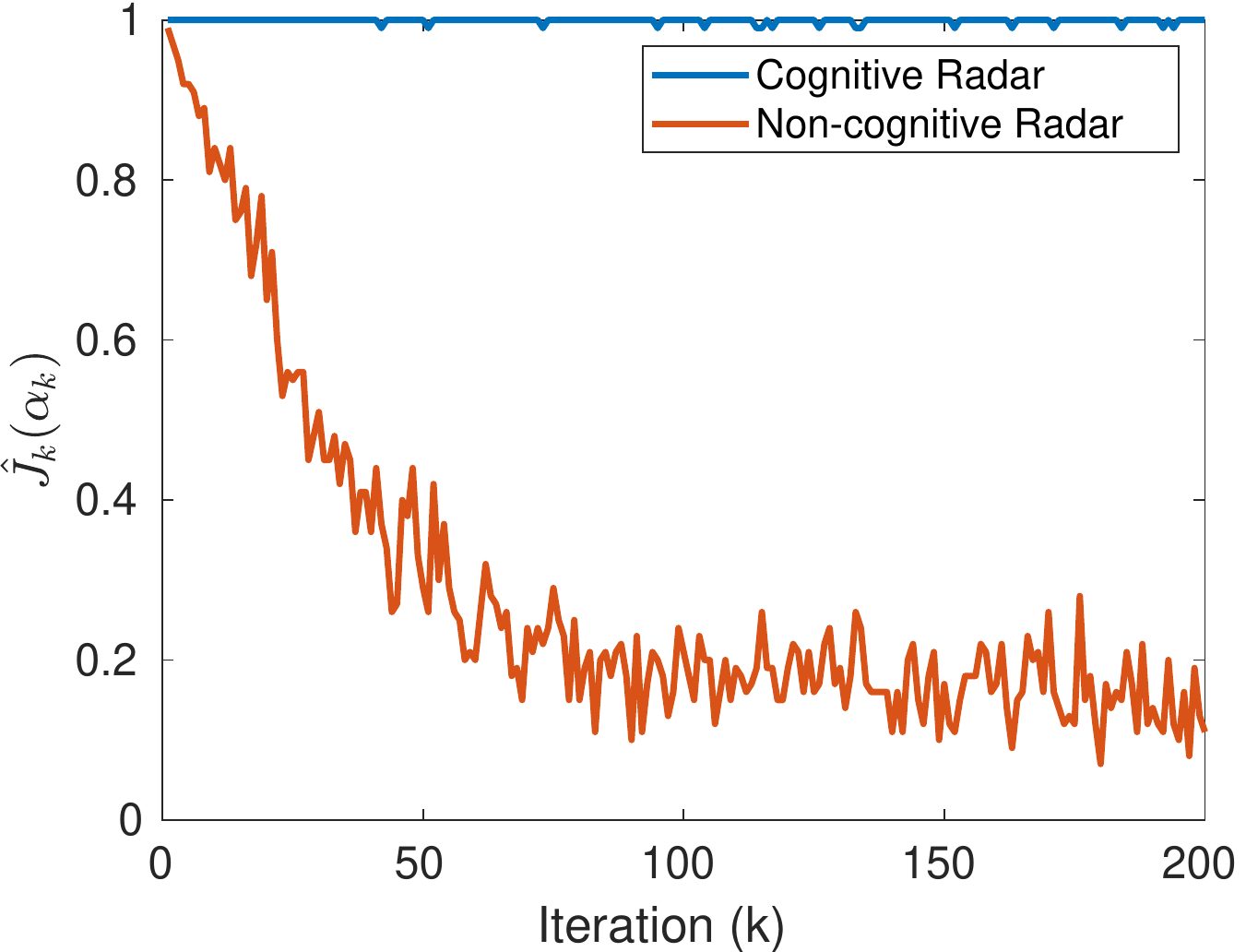}
  \caption{Performance of SPSA Algorithm \ref{alg:spsa}. Given the radar response in beam scheduling amongst 3 targets, we adaptively optimize our probe signal to minimize the probability of Type-II errors in  the detector, namely (\ref{eqn: SPSA Objective}). $\hat{J}_k(\Probe_k)$ is the empirical Type-II error probability (\ref{eqn:SPSACost}).}
  \label{fig:spsa}
\end{figure}

 \section{Conclusion and Discussion}
 Cognitive radars adapt their sensing  by  optimizing their  waveform, aperture and beam allocation.
 The main idea of this paper   was to formulate a revealed preference framework (from microeconomics)  to detect such constrained utility maximization behavior in radars.  As mentioned in the introduction, such methods generalize classical inverse reinforcement learning. The main results of the paper are:
 \\
 (i) Spectral revealed preferences algorithms to detect if a radar is optimizing its waveform.
Our probe input comprises purposeful maneuvers that modulate the spectrum (vector of eigenvalues) of the state noise covariance matrix. The radar responds with an optimized waveform which modulates the spectrum of the observation noise covariance matrix.
 The spectra of  the state  and observation noise covariance matrices were  used in Afriat's theorem to detect utility maximization behavior in radar waveforms. A generalization involving nonlinear budgets and the algebraic Riccati equation was obtained. We presented
similar methods to detect if a radar is optimizing a utility function when it allocates its  beam among multiple targets. \\
 (ii) We then developed stochastic revealed preference tests when  either the enemy radar's response is observed by us in noise, or the enemy radar observes our input in noise.  Specifically, we developed  a statistical hypothesis test to detect utility maximization behavior by the radar. We gave tight bounds for the Type-I error of the detector. \\
 (iii) Finally we presented  an SPSA based stochastic optimization algorithm to adaptively interrogate the enemy  radar to detect if it is cognitive. The algorithm minimizes the Type-II detection error subject to constraints on the Type-I error.

 {\bf Extensions}.  This paper focused on detecting radars that adapt waveforms to improve \textit{tracking}. The ideas  can be extended  to radars which adapt waveforms to improve \textit{detection} (in clutter and jamming). In this paper,  our probing of the enemy radar was performed  via purposeful maneuvers by  modulating our  state covariance matrix  $\snoisecov$.  To extend the result to the  latter case,  our probing of the enemy radar will involve  emitting certain classes of signals and  modifying reflected signals so that we can ascertain how the radar changes its waveform to improve detectability.

Finally, the methodology in this  paper is an early step  in understanding
how to design stealthy cognitive radars whose cognitive functionality is difficult to detect by an observer. In future work we will consider \textit{how to design a  smart radar that acts dumb?}. This  generalizes the physics based low-probability of intercept (LPI) requirement of radar (which requires low power emission) to the    systems-level issue: How should the radar choose its actions in order to avoid detection of its cognition?

\bibliographystyle{IEEEtran}

\bibliography{$HOME/texstuff/styles/bib/vkm}
\end{document}